\documentclass{article}
\usepackage[utf8x,latin1]{inputenc}

\usepackage{amsmath}
\usepackage{amsthm}
\usepackage{amsfonts}
\usepackage{amssymb}
\usepackage{amsbsy}

\usepackage{graphicx}

\usepackage{color}
\newcommand{\ifthesis}[1]{}

\newcommand{\ifgermanbookchapter}[1]{}

\newcommand{\ifarxiv}[1]{#1}
\theoremstyle{definition}
\newtheorem{definition}{Definition}[section]
\newcommand{\chapter}[1]{}

\newcommand{\ifbookchapter}[1]{}

\theoremstyle{definition}

\theoremstyle{plain}
\newtheorem{theorem}[definition]{Theorem}
\newtheorem{lemma}[definition]{Lemma}
\newtheorem{proposition}[definition]{Proposition}
\newtheorem{corollary}[definition]{Corollary}

\theoremstyle{remark}
\newtheorem{remark}[definition]{Remark}
\newtheorem{example}[definition]{Example}

\DeclareRobustCommand{\r}{{\mathbb R}}

\newcommand{\R}{{\mathbb R}}  
\newcommand{\Z}{{\mathbb Z}}  

\newcommand{\be}[1]{\begin{equation}\label{#1}}
\newcommand{\ee}{\end{equation}}
\newcommand{\bl}[1]{\begin{lemma}\label{#1}}
\newcommand{\ble}[1]{\begin{lemmaex}\label{#1}}
\newcommand{\br}[1]{\begin{remark}\label{#1}}
\newcommand{\bt}[1]{\begin{theorem}\label{#1}}
\newcommand{\bd}[1]{\begin{definition}\label{#1}}
\newcommand{\bp}[1]{\begin{proposition}\label{#1}}
\newcommand{\bc}[1]{\begin{corollary}\label{#1}}
\newcommand{\bex}[1]{\begin{example}\label{#1}}
\newcommand{\ec}{\mybox\end{corollary}}
\newcommand{\eex}{\mydiamond\end{example}}
\newcommand{\eem}{\mydiamond\end{example}}
\newcommand{\el}{\mybox\end{lemma}}
\newcommand{\er}{\mybox\end{remark}}
\newcommand{\et}{\qed\end{theorem}}
\newcommand{\ed}{\mytriangle\end{definition}}
\newcommand{\ep}{\mybox\end{proposition}}
\newcommand{\epr}{\end{proof}}
\newcommand{\bpr}{\begin{proof}}

\newcommand{\ecs}{\end{corollary}}
\newcommand{\eexs}{\mydiamond\end{example}}
\newcommand{\els}{\end{lemma}}
\newcommand{\ers}{\end{remark}}
\newcommand{\ets}{\end{theorem}}
\newcommand{\eds}{\end{definition}}
\newcommand{\eps}{\end{proposition}}
\newcommand{\mybox}{\hfill $\Box$} 
\newcommand{\beq}{\begin{eqnarray}}
\newcommand{\eeq}{\end{eqnarray}}
\newcommand{\beqn}{\begin{eqnarray*}}
\newcommand{\eeqn}{\end{eqnarray*}}
\newcommand{\bi}{\begin{itemize}}
\newcommand{\ei}{\end{itemize}}
\newcommand{\ben}{\begin{enumerate}}
\newcommand{\een}{\end{enumerate}}

\newcommand{\stochQ}{{\cal S}_\theta^X}
\newcommand{\stochU}{{\cal S}_\theta^U}

\newcommand{\Eq}{{\cal E}}  
\newcommand{\Oeq}{{\cal K}}  
\DeclareMathOperator{\e}{e}
\DeclareRobustCommand{\thetaminus}[1]{\theta_{-#1}}
\DeclareRobustCommand{\temperedU}{U_\theta^\Omega}
\DeclareRobustCommand{\temperedX}{X_\theta^\Omega}
\DeclareRobustCommand{\boundedU}{U_\infty^\Omega}
\DeclareRobustCommand{\boundedX}{X_\infty^\Omega}
\DeclareRobustCommand{\n}{{\mathbb N}}
\DeclareRobustCommand{\z}{{\mathbb Z}}

\DeclareRobustCommand{\r}{{\mathbb R}}

\DeclareRobustCommand{\realp}{\mathbb P}
\DeclareRobustCommand{\Time}{{\cal T}}
\DeclareRobustCommand{\Tplus}{{\cal T}_{\geqslant 0}}
\DeclareRobustCommand{\borel}{{\cal B}}
\DeclareRobustCommand{\F}{{\cal F}}
  \DeclareRobustCommand{\calF}{\F}
\DeclareRobustCommand{\calK}{{\cal K}}
\DeclareRobustCommand{\calU}{{\cal U}}
\DeclareRobustCommand{\calS}{{\cal S}}

\DeclareRobustCommand{\omegaxb}{X_\borel^\Omega}
\DeclareRobustCommand{\UBOmega}{U_\borel^\Omega}
  \DeclareRobustCommand{\omegaub}{\UBOmega}
\DeclareRobustCommand{\omegaxtheta}{X_\theta^\Omega}
\DeclareRobustCommand{\omegautheta}{U_\theta^\Omega}
\DeclareRobustCommand{\linftyofX}{X_\infty^\Omega}
\DeclareRobustCommand{\linftyofU}{U_\infty^\Omega}

\DeclareRobustCommand{\sthetau}{{\cal S}^{U}_\theta}
\DeclareRobustCommand{\sthetay}{{\cal S}^{Y}_\theta}
\DeclareRobustCommand{\soU}{{\cal S}_c^U}

\DeclareRobustCommand{\sinftyU}{{\cal S}_\infty^U}

\DeclareRobustCommand{\sthetauconst}{{\bar{\cal S}^U}_\theta}
\DeclareRobustCommand{\temparrow}{\rightarrow_\theta}
\DeclareRobustCommand{\longtemparrow}{\longrightarrow_\theta}
\newcommand{\as}{\quad \text{as}\quad}
\newcommand{\mytriangle}{\hfill $\triangle$}
\newcommand{\mydiamond}{\hfill $\Diamond$}
\newcommand{\comment}[1]{}
\newcommand{\ifonecolumn}[1]{#1}
\newcommand{\iftwocolumn}[1]{}
\begin{document}

\title{Remarks on random dynamical systems
       with inputs and outputs,\\ and a small-gain theorem for monotone RDS}
\author{Michael Marcondes de Freitas and Eduardo D. Sontag\\
Mathematics Department\\
Rutgers University, Piscataway, NJ, USA}
\maketitle

\begin{abstract}
This note introduces a new notion of random dynamical systems with inputs and
outputs, and sketches a small-gain theorem for monotone systems which
generalizes a similar theorem known for deterministic systems.
\end{abstract}

\section{Introduction}

\emph{Monotone systems}, whose mathematical development was pioneered by
Hirsch~\cite{Hirsch,Hirsch2} and Smith~\cite{smith,Hirsch-Smith}, play a key
role in many application areas, and particularly in the modeling and analysis
of biological systems.  They constitute a class of dynamical systems for which
a rich theory exists, endowing them with very robust dynamical
characteristics.  In order to analyze interconnections of monotone systems,
however, it is necessary to extend the notion and introduce the concept of
\emph{monotone systems with inputs and outputs}, as standard in control
theory~\cite{mct}, so as to incorporate input and output channels.  
This was done in~\cite{monotoneTAC}, which also presented a result that
guarantees stability of monotone systems under negative feedback (so that the
closed-loop system is no longer monotone); this result may be viewed as a
``small gain theorem'' in terms of an appropriate notion of system gain.

An interesting question is: to what extent do results for monotone I/O systems
extend to the analysis of systems subject to stochastic uncertainty or random
inputs?  This paper proposes an approach based upon the notion of ``random
dynamical system'' (RDS) due to Arnold~\cite{arnold-2010}, and more
specifically the subclass of \emph{monotone} RDS studied
in~\cite{chueshov-2002}.
Even more than when passing from deterministic monotone systems to monotone
systems with inputs and outputs, the generalization is not entirely
straightforward, and many subtle mathematical and conceptual details have to
be worked out.  We introduce the necessary formalism in this note (in a more
general context of not necessarily monotone systems) and
sketch an analogue
of the small gain theorem proved in~\cite{monotoneTAC}.

\section{Random Dynamical Systems}
We first review the random dynamical systems framework of Arnold
\cite{arnold-2010}. Along the way we introduce a couple of pieces of terminology
not found in \cite{arnold-2010} to facilitate the discussion. Suppose given a 
\emph{measure preserving dynamical system}%
\footnote{Arnold \cite[page 635]{arnold-2010} and Chueshov \cite[Definition
1.1.1 on page 10]{chueshov-2002} refer to such an object primarily as a 
\emph{metric dynamical system}. We find {\em measure preserving}, which Arnold
also uses as a synonym, less confusing and more informative.} ({\em MPDS}\,)
$$\theta =(\Omega ,{\cal F},\realp,\{\theta _t\}_{t \in \Time})\,;$$ that is, a
probability space $(\Omega, {\cal F}, \realp)$, a topological group $(\Time,+)$,
 and a measurable flow $\{\theta_t\}_{t \in \Time}$ of measure preserving maps
$\Omega \rightarrow \Omega$ satisfying (T1)--(T3):
\begin{enumerate}
\item[(T1)] $(t,\omega) \mapsto \theta_t\omega$, $(t,\omega) \in \Time \times
\Omega$, is ($\borel(\Time)\otimes{\cal F}$)-measurable,

\item[(T2)] $\theta_{t+s} = \theta_t \circ \theta_s$ for every $t,s \in \Time$
(semigroup property),

\item[(T3)] $\realp \circ \theta_t = \realp$ for each $t \in \Time$ (measure
preserving\footnote{
Property (T3) is normally \cite[Definition 1.1]{walters-2000} stated as
$$\realp(\theta_t^{-1}(B)) = \realp(B)\,, \quad \forall B \in {\cal F}\,,\
\forall t \in \Time\,.$$ But since it follows from (T2) that $\theta_t$ is
invertible with $\theta_t^{-1} = \theta_{-t}$ for each $t \in \Time$, the two
formulations are equivalent in this context.
}). 
\end{enumerate}
In this work $\Time$ will always refer to either $\R$ or $\Z$, depending on
whether
one is
talking about continuous or discrete time, respectively. In either case
$\Tplus$ refers to the nonnegative elements of $\Time$. We will occasionally
need to make measure-theoretic considerations about $\Time$ or Borel subsets of
it. If $\Time = \r$, that is, in continuous time, then we tacitly equip any
Borel subset of $\Time$ with the measure induced by the Lebesgue measure on
$\r$. If $\Time = \z$, or in discrete time, then we think of the counting
measure in $\z$. When $\Time = \z$, it follows from (T2) that $\theta$ is
completely determined by $\theta_1 = \theta(1,\cdot)$. In that case we will
abuse the notation and use the same $\theta$ to denote both the underlying MPDS
and $\theta_1$.

In the context of a given MPDS $\theta$, a set $B \in \calF$ is said to be
$\theta$-invariant if $\theta_t(B) = B$ for all $t \in \Time$. We say that an
MPDS $\theta$ is {\em ergodic {\em(}under $\realp$\,{\em)}} if, whenever $B
\in \calF$ is $\theta$-invariant, then we have either $\realp(B) = 0$ or
$\realp(B) = 1$.

 Let $X$ be a metric space constituting the measurable space $(X, \borel)$ when
equipped with the $\sigma$-algebra $\borel$ of Borel
subsets of $X$. A {\em {\em (}continuous\,{\em )} random dynamical system}
({\em RDS}\,)
{\em on $X$} is a pair $(\theta,\varphi)$ in which $\theta$ is an MPDS and
$$\varphi :
\Tplus \times
\Omega \times X \longrightarrow X$$ is a ({\em continuous}\,) {\em cocycle over
$\theta$}; that is,
a ($\borel({\Tplus}) \otimes \calF \otimes \borel$)-measurable map such that 
\begin{enumerate}
\item[(S1)] $\varphi(t,\omega) := \varphi(t,\omega,\cdot) : X \rightarrow X$ is
continuous for each $t \in \Tplus$, $\omega \in \Omega$,

\item[(S2)] $\varphi(0,w) = id_X$ for each $\omega \in \Omega$, and (cocycle
property)
\[
\varphi(t+s, \omega) = \varphi(t, \theta_s\omega) \circ \varphi(s,
\omega)\,,\quad \forall s,t \in \Tplus\,,\ \forall \omega \in \Omega\,.
\]
\end{enumerate}
The cocycle property generalizes the semigroup
property of deterministic dynamical systems. Thus RDS's include
deterministic dynamical systems as the special case in
which $\Omega$ is a singleton.

\subsection{Trajectories, Equilibria, and $\theta$-Stationary Processes}

\newcommand{\omegaqb}{X_{\borel}^{\Omega}}
In the context of RDS's, the analogue to points in the state space $X$ for a
deterministic
system are random variables $\Omega \rightarrow X$, that is,
$\borel$-measurable maps $\Omega \rightarrow X$.
We denote the set of all random variables on a metric space $X$ by
$\omegaqb$. We refer to a ($\borel(\Tplus)\otimes\calF$)-measurable map $q:
\Tplus \times \Omega
\rightarrow X$ as a $\theta${\em -stochastic process\footnote{A
``$\theta$-stochastic process'' is indeed a stochastic process in the
traditional sense. We use the prefix ``$\theta$-'' to emphasize the underlying
probability space, as well as the time semigroup.} on} $X$, and
denote $q_t := q(t,\cdot)$ for each $t \in \Tplus$. The set of all
$\theta$-stochastic processes on a metric space $X$ is denoted by $\stochQ$.

Let $(\theta,\varphi)$ be an RDS. Given $x \in \omegaqb$, we define the ({\em
forward}\,) {\em trajectory starting at} $x$ to be the $\theta$-stochastic
process
$\xi^{x} \in \stochQ$ defined by
\begin{equation}\label{eq:trajectory}
\xi^x_t(\omega) := \varphi(t,\omega ,x(\omega))\,, \quad (t, \omega) \in \Tplus
\times \Omega\,.
\end{equation}
The {\em pullback trajectory starting at} $x$ is in turn defined to be the
$\theta$-stochastic process $\check\xi^{x}: \Tplus \times \Omega \rightarrow X$
defined by
\begin{equation}\label{eq:pullbacktrajectory}
\check\xi^x_t(\omega) := \varphi(t,\thetaminus{t}\omega
,x(\thetaminus{t}\omega))\,, \quad (t, \omega) \in \Tplus \times \Omega\,.
\end{equation}
More generally, the {\em pullback} of a $\theta$-stochastic process $q \in
\stochQ$ is the
$\theta$-stochastic process $\check q \in \stochQ$ defined by
\[
 \check q_t(\omega) := q_t(\thetaminus{t}\omega)\,, \quad (t,\omega) \in \Tplus
\times \Omega\,.
\]
So the pullback trajectory starting at $x$ is simply the pullback of the
forward trajectory starting at $x$. We will always use the accent
$\phantom{.}\check{}\phantom{.}$ to indicate the pullback of the
$\theta$-stochastic process being accented.

We slightly modify the standard notion of equilibrium for RDS's (see, for
instance, \cite[Definition 1.7.1 on page 38]{chueshov-2002}) to allow for the
defining property to hold only almost everywhere, as opposed to everywhere. So
an \emph{equilibrium} of an RDS $(\theta,\varphi)$ is a random variable $x \in
\omegaxb$ such that
\[
\xi^x_t(\omega) = \varphi(t,\omega ,x(\omega)) = x(\theta _t\omega )\,, \quad
\forall t \in \Tplus\,,\ \forall \omega \in \widetilde{\Omega}\,,
\]
for some $\theta$-invariant $\widetilde{\Omega} \subseteq \Omega$ of full
measure\footnote{That is,
$\theta_t\widetilde{\Omega} = \widetilde{\Omega}$ for all $t \in \Time$, and
$\realp(\widetilde{\Omega}) = 1$.}. It is often not necessary to specify the
said
$\widetilde{\Omega}$. So we say ``for $\theta$-almost all $\omega \in \Omega$''
and write
\[
 \text{``}\ \widetilde{\forall} \omega \in \Omega \ \text{''}
\]
to mean `for all $\omega \in \widetilde{\Omega}$, for some $\theta$-invariant
set $\widetilde{\Omega} \subseteq \Omega$ of full measure'.
 
In view of the notion of pullback convergence with which we will be working
(see Subsection \ref{subsec:convergence}), it is more natural to think of
the concept of equilibrium in terms of pullback trajectories. Observe that a
random
variable $x \in \omegaxb$ is an equilibrium of the RDS $(\theta,\varphi)$ if,
and only
if
\[
 \check\xi^x_t(\omega) = \varphi(t,\thetaminus{t}\omega
,x(\thetaminus{t}\omega)) = x(\omega )\,, \quad \forall t \in \Tplus\,,\
\widetilde{\forall} \omega \in \Omega\,.
\]

The remaining of this section is devoted to interpreting the concept of
equilibrium for an RDS in terms of a shift operator in the set $\stochQ$ of all
$\theta$-stochastic processes on $X$. For each $s\in \Tplus$, let
\be{eq:rho}
\begin{array}{rcl}
\rho _s : \stochQ &\longrightarrow  &\stochQ\\
 q &\longmapsto  &\rho _s(q)
\end{array}
\ee
be defined by
\be{eq:rhodef}
(\rho _s(q))_t(\omega) := q_{t+s}(\theta _{-s}\omega )\,, \quad (t,\omega) \in
\Tplus \times \Omega\,.
\ee

\begin{definition}\label{def:q-constant}
A $\theta $-stochastic process $\bar q \in \stochQ$ is said to be \emph{$\theta
$-stationary}
if
\begin{equation*}
(\rho _s(\bar q))_t(\omega)=\bar q_t(\omega)\,, 
\end{equation*}
for all $s,t \in \Tplus$, for $\theta$-almost all $\omega \in \Omega$.
\mytriangle
\end{definition}
We use the preffix ``$\theta$-'' in ``$\theta$-stationary'' to emphasize the
dependence on the underlying MPDS $\theta$.

\bl{lemma:q-constant}\label{lemma:theta-constant}
The $\theta $-stochastic process $\bar q \in \stochQ$ is $\theta $-stationary if
and only if
there exists a random variable ${q} \in \omegaxb$ such that
\begin{equation}\label{eq:theta-constant}
\bar q_t(\omega ) = {q}(\theta _t\omega )\,,
\quad
\forall t\in \Tplus\,,\ \widetilde{\forall} \omega \in \Omega \,.
\end{equation}
\els

\bpr (Sufficiency)
Suppose that (\ref{eq:theta-constant}) holds for some ${q} \in
\omegaqb$.
Pick any $s\in \Tplus$.
For any $t \in \Tplus$ and $\theta$-almost all $\omega \in \Omega$,
\[
(\rho _s(\bar q))_t(\omega ) = \bar q_{t+s}(\theta _{-s}\omega )
              = {q}(\theta _{t+s}\theta _{-s}\omega )
              = {q}(\theta _t\omega )
              = \bar q_t(\omega) \,.
\]
So $\bar q$ is $\theta$-stationary.

(Necessity)
Suppose that $\bar q \in \stochQ$ is $\theta$-stationary and define ${q} \in
\omegaqb$ by                                   
\begin{equation}\label{eq:dcc}
{q}(\omega ) := \bar q_0(\omega )\,, \quad \omega \in \Omega\,.
\end{equation}
We have
\[
\bar q_{t+s}(\theta _{-s}\omega ) = (\rho_{s}(q))_t(\omega) = \bar q_t(\omega
)\,,\quad \forall s,t \in
\Tplus\,,\ \widetilde{\forall} \omega \in \Omega\,.
\]
Setting $t=0$ and renaming $s$ as $t$ we then have $$\bar q_t(\theta
_{-t}\hat\omega) = \bar q_0(\hat\omega ) = {q}(\hat\omega )\,,\quad \forall t
\in \Tplus\,,\ \widetilde{\forall}\hat\omega \in \Omega\,.$$
Given any $\omega \in \widetilde{\Omega} $ and any $t\in \Tplus$, we may apply
this property with
$\hat\omega =\theta _t\omega $ due to the $\theta$-invariance of
$\widetilde{\Omega}$, thus obtaining
\[
\bar q_t(\omega ) = {q}(\theta _t\omega )\,.
\]
Therefore (\ref{eq:theta-constant}) holds.
\epr

Note that the random variable ${q}$ associated to $\bar q$ is unique up to a
$\theta$-invariant set of measure zero. Indeed, it is determined
$\theta$-almost everywhere by Equation (\ref{eq:dcc}). Thus, we have:

\begin{corollary}\label{lem:equilstate}
Given an RDS $(\theta,\varphi)$ over a metric space $X$ and a random state $x
\in \omegaqb$,
the following two properties are equivalent:
\begin{itemize}
\item[{\em (1)}] $x$ is an equilibrium;
\item[{\em (2)}] the trajectory $\xi^x$, as defined in {\em
Equation (\ref{eq:trajectory})}, is
$\theta$-stationary;
\item[{\em (3)}] the map $t \mapsto \check\xi_t^x \in
\omegaqb$, $t \in \Tplus$, is constant.

\end{itemize}
\end{corollary}

We will always use an overbar to denote the $\theta$-stationary
$\theta$-stochastic process $\bar q$ associated with a given random variable
$q$.

\ifarxiv{%
\subsection{Perfection of Crude Cocycles}
We briefly review the theory of perfection of crude cocycles discussed in
Arnold's \cite[Section 1.2]{arnold-2010}. It is customary for the definition of
an RDS to require that the cocycle property of $\varphi$ in (S2) holds for every
$s,t \in \Tplus$ and {\bf every} $\omega \in \Omega$. If we want to
emphasize this fact we shall say that $\varphi$ is a {\em perfect cocycle} (over
the underlying MPDS $\theta$).

\begin{definition}[Crude Cocycle]\label{def:crudecocycle}
We say that $\varphi\colon \Tplus \times \Omega \times X \rightarrow X$ is a
{\em crude cocycle {\em (}over $\theta$\,{\em )}} if it is a
$(\borel(\Time)\otimes\calF\otimes\borel)$-measurable map satisfying (S1) and
\begin{itemize}
\item[(S2$^\prime$)] $\varphi(0,w) = id_X$ for each $\omega \in
\Omega$, and for every $s \in \Tplus$, there exists a subset $\Omega_s \subseteq
\Omega$ of
full measure such that
\[
\varphi(t+s, \omega) = \varphi(t, \theta_s\omega) \circ \varphi(s,
\omega)\,,\quad \forall t \in \Tplus\,,\ \forall \omega \in \Omega_s\,.
\]
\end{itemize}
The $\Omega_s$'s need not be $\theta$-invariant.
\mytriangle
\end{definition}

As Arnold points out, there are circumstances where this flexibility in the
requirements for a cocycle is desirable. For instance, the flow of a stochastic
differential equation is only guaranteed to be a crude cocycle \cite[Section
2.3]{arnold-2010}. Another example will come up below after we introduce random
dynamical systems with inputs. Consider (deterministic) controlled dynamical
systems. Such systems yield a (deterministic) dynamical system when restricted
to a constant input. One would expect a sensible extension of the concept to
random dynamical systems to have an analogous property. However we shall see in
the proof of Lemma \ref{lemm:rdsi-constant} in the next section that the
restriction of the flow of an RDS with inputs to a $\theta$-stationary input is
not necessarily a perfect cocycle.

\begin{definition}[Indistinguishable Crude Cocycles]
 Let $\theta$ be an MPDS and $\varphi,\psi\colon \Tplus \times \Omega \times X
\rightarrow X$ crude cocycles over $\theta$. If there exists a subset $N \in
\calF$ such that
$\realp(N) = 0$ and
\[
 \{\omega \in \Omega\,;\ \varphi(t,\omega) \neq \psi(t,\omega)\,,\ \text{for
some}\ t \in \Tplus\} \subseteq N\,,
\]
then $\varphi$ and $\psi$ are said to be {\em indistinguishable}.
\mytriangle
\end{definition}

\begin{definition}[Perfection of Crude Cocycles]\label{def:perfection}
 A {\em perfection} of a crude cocycle $\varphi$ is a perfect cocycle $\psi$
(with respect to the same MPDS and evolving in the same state space) such that
$\varphi$ and $\psi$ are indistinguishable. In this case we say that $\varphi$
is {\em perfected by $\psi$}, or that $\psi$ is a {\em perfection of $\varphi$}.
\mytriangle
\end{definition}

In this work we will not need the full power of Arnold's theory of perfection of
crude cocycles. The $\Omega_s$'s of the crude cocycles we shall have to deal
with will be actually $\theta$-invariant, and so it will be enough to simply
redefine the flow on a $\theta$-invariant set of measure zero. We nevertheless
state the proposition below for the sake of completeness.

\begin{proposition}\label{thm:perfection}
Let $\theta = (\calF,\Omega,\realp,(\theta_t)_{t \in \Time})$ be an MPDS with
$\Time = \z$ or $\Time = \r$. Suppose $\varphi\colon \Tplus \times \Omega \times
X \rightarrow X$ is a crude cocycle over $\theta$ evolving in a locally compact,
locally connected, Hausdorff topological space $X$. Then $\varphi$ can be
perfected; in other words, there exists a perfect cocycle $\psi\colon \Tplus
\times \Omega \times X \rightarrow X$ such that $\varphi$ and $\psi$ are
indistinguishable.
\end{proposition}

\begin{proof}
See Arnold \cite[Theorem 1.2.1]{arnold-2010} for the discrete case, which
actually holds with weaker hypotheses and yields stronger conclusions. For the
continuous case, see Arnold \cite[Theorem 1.2.2 and Corollary
1.2.4]{arnold-2010}.
\end{proof}%
}%
\subsection{Pullback Convergence}\label{subsec:convergence}
We work with the notion of pullback convergence developed in the literature
and canonized in the works of Arnold and Chueshov
\cite{arnold-2010,chueshov-2002}.

\begin{definition}(Pullback Convergence)
 A $\theta$-stochastic process $\xi \in \stochQ$ is
said
to {\em converge to a random variable $\xi_\infty \in \omegaqb$ in the
pullback sense} if
\[
 \check\xi_t(\omega)
 := \xi_t(\theta_{-t}\omega) \longrightarrow \xi_\infty(\omega) \quad
\text{as}\,\quad t \rightarrow \infty\,,
\]
for $\theta$-almost all $\omega \in \Omega$. 
\mytriangle
\end{definition}

\begin{proposition}\label{prop:convergencetoeq}
 Let $(\theta,\varphi)$ be an RDS evolving on a metric space $X$. Suppose there
exists a random initial state
$x \in
X^\Omega_{\borel}$ and a map ${x}_\infty: \Omega
\rightarrow X$ such that
\begin{equation}\label{eq:pullbacklimit}
 \check\xi^{x}_t(\omega) =
\varphi(t,\thetaminus{t}\omega,x(\thetaminus{t}\omega)) \longrightarrow
{x}_\infty(\omega)\quad \text{as}\quad t \rightarrow \infty\,,\quad
\widetilde{\forall} \omega \in \Omega\,.
\end{equation}
Then ${x}_\infty$ is an equilibrium.
\end{proposition}

\begin{proof}
 For each $t \in \Tplus$, the map $\omega \mapsto
\varphi(t,\thetaminus{t}\omega,x(\thetaminus{t}\omega))$, $\omega \in \Omega$,
is measurable, since it is the composition of measurable maps:
\[
 \omega \longmapsto \thetaminus{t}\omega \longmapsto
x(\thetaminus{t}\omega)\,,
\]
\[
 (\thetaminus{t}\omega, x(\thetaminus{t}\omega)) \longmapsto
\varphi(t,\thetaminus{t}\omega, x(\thetaminus{t}\omega))\,.
\]
So it follows from \cite[Chapter 11, \textsection 1, Property M7 on page
248]{lang-1983} that ${x}_\infty$ is measurable. (If $\Time$ is continuous time,
just pick a subsequence $(t_n)_{n\in\n}$ going to infinity.)

In addition, for each $\omega \in \Omega$ such that the limit in
(\ref{eq:pullbacklimit}) exists, and each $\tau \in \Tplus$, we have
\[
 \lim_{t\to\infty}\varphi(t-\tau,\theta_{\tau-t}\omega,x(\theta_{\tau-t}
\omega)) = {x}_\infty(\omega)
\]
also. By $\theta$-invariance, the limit in (\ref{eq:pullbacklimit}) exists
for $\theta_\tau\omega$ as well. Hence
\[
 \begin{array}{rcl}
  {x}_\infty(\theta_\tau\omega)
&= &\displaystyle\lim_{t\to\infty}
\varphi(t,\theta_{-t}\theta_\tau\omega,x(\theta_{-t}\theta_\tau\omega)) \\[1ex]
\ifonecolumn{&=&}\iftwocolumn{&&\adjusttoleft =\;\;}
\displaystyle\lim_{t\to\infty} \varphi(\tau,\theta_{t - \tau} \theta{\tau
-t}\omega,\varphi(t-\tau,\theta_{\tau-t}\omega,x(\theta_{\tau-t}\omega)))
\\[1ex]
\ifonecolumn{&=&}\iftwocolumn{&&\adjusttoleft =\;\;}
\varphi(\tau,\omega,{x}_\infty(\omega))
 \end{array}
\]
by continuity (property (S1) in the definition of an RDS).
\end{proof}
\section{RDS's with Inputs and Outputs}
We now define a new concept.  It extends the notion of RDS's to systems in which
there is an external input or forcing function.  A contribution of this work
is the precise formulation of this concept, particularly the way in which the
argument of the input is shifted in the semigroup (cocycle) property.

As in the previous section, given a metric space $U$, we equip it with its Borel
$\sigma$-algebra $\borel(U)$ and denote by $\omegaub$ the set of Borel
measurable maps $\Omega \rightarrow U$. Let $\sthetau$ be the set of all
$\theta$-stochastic processes $\Tplus \times \Omega \rightarrow U$. Given $u,v
\in \sthetau$ and $s \in \Tplus$, we define $u \lozenge_s v\colon \Tplus \times
\Omega \rightarrow U$ by
\[
(u \lozenge_s v)_\tau(\omega) = \left\{
\begin{array}{rl}
u_\tau(\omega)\,,\ &0 \leqslant \tau < s\\
v_{\tau - s}(\theta_s\omega)\,,\ &s \leqslant \tau
\end{array}
\right.\,, 
\ifonecolumn{\quad \tau \in \Tplus\,,\ \omega \in \Omega\,.}
\]
\iftwocolumn{for all $\tau \in \Tplus$, $\omega \in \Omega$.}%
Given $\tilde u \in U$, we denote by $c(\tilde u)$ the trivial
$\theta$-stochastic process defined by $(c(\tilde u))_t(\omega) := \tilde u$ for
every $t \in \Tplus$ and every $\omega \in \Omega$.

\begin{definition}[$\theta$-Inputs]
We say that a subset ${\calU} \subseteq \sthetau$ is a {\em set of
$\theta$-inputs} if $c(\tilde u) \in \calU$, for every $\tilde u \in U$, and
$u\lozenge_s v \in \calU$, for any $u,v \in \calU$ and any
$s \in \Tplus$. In other words, a set of $\theta$-inputs is a subset of
$\sthetau$ which contains all constant inputs and is closed under concatenation.
\mytriangle
\end{definition}

\begin{definition}\label{def:RDSI}
A \emph{random dynamical system with inputs {\em (}RDSI\,{\em )}} is a triple
$(\theta
,\varphi,\calU)$
consisting of an MPDS
$\theta =(\Omega ,{\cal F},{\mathbb P},\{\theta _t\}_{t \in \Time})$, a set of
$\theta$-inputs $\calU \subseteq \sthetau$, and a map
\[
\varphi:\Tplus\times \Omega \times X\times \calU \rightarrow X
\]
satisfying
\ben
\item[(I1)] $\varphi(\cdot,\cdot,\cdot,u)\colon \Tplus \times \Omega \times X
\rightarrow X$ is ($\borel(\Tplus)\otimes\calF\otimes\borel$)-measurable for
each fixed $u \in \calU$;
\item[(I1$^\prime$)]
the map $\widetilde\varphi\colon \Tplus \times \Omega \times X \times U
\rightarrow X$ defined by
\[
  \widetilde\varphi(t,\omega,x,\tilde u) := \varphi(t,\omega,x,c(\tilde u))\,,
\quad (t,\omega,x,\tilde u) \in \Tplus \times \Omega \times X \times U\,,
\]
is ($\borel(\Tplus)\otimes \calF \otimes \borel
\otimes
\borel(U)$)-measurable;
\item[(I2)]
$\varphi(t,\omega ,\cdot ,u ): X \rightarrow X$ is continuous, for each fixed
$(t,
\omega, u) \in \Tplus \times \Omega \times \calU$;
\item[(I3)]
$\varphi(0,\omega ,x,u) = x$ for each $(\omega , x, u) \in \Omega \times X
\times \calU$;
\item[(I4)]
given $s,t\in \Tplus$, $\omega \in \Omega $, $x\in X$, $u, v \in \calU$,
if
$$\varphi(s,\omega,x,u) = y$$ and $$\varphi(t, \theta_s\omega, y, v) = z\,,$$
then $$\varphi(s+t, \omega, x, u\lozenge_sv) = z\,;$$
\item[(I5)] and given $t \in \Tplus$, $\omega \in \Omega$, $x \in X$, and $u,v
\in
\calU$, if $u_\tau(\omega) = v_\tau(\omega)$ for almost all $\tau \in [0,t)$,
then $\varphi(t,\omega,x,u) = \varphi(t,\omega,x,v)$.
\een
We refer to the elements $u\in \calU$ as \emph{$\theta$-inputs}, or simply
\emph{inputs}. Whenever we talk
about an RDSI $(\theta,\varphi,\calU)$, we tacitly assume the notation laid
above,
unless otherwise specified.
\mytriangle
\end{definition}

\noindent (I1), (I1$^\prime$) and (I2) are regularity conditions. (I3) means
that nothing
has
``happened'' if one is still at time $t = 0$. (I4) generalizes the cocycle
property and (I5) states that the evolution of an RDS subject to an input $u$
is, so to speak, independent of ``irrelevant'' random input values.

\begin{remark}\label{rem:rhoshift}
Notice that for each $s,t\in \Tplus$, $x\in X$, $\omega \in \Omega $,
\[
\varphi(t+s,\omega ,x,u) \;=\; \varphi(t,\theta _s\omega ,\varphi(s,\omega
,x,u),\rho _s(u))\,, \quad \forall u \in \calU\,,
\]
where $\rho_s\colon \sthetau \rightarrow \sthetau$ is defined by Equation
(\ref{eq:rhodef})\footnote{We will use
the same notation $\rho_s$ for the shift operator $\calS_\theta^V \rightarrow
\calS_\theta^V$ defined by Equation (\ref{eq:rhodef}), irrespective of the
underlying metric space $V$. Since the domain of any $\theta$-stochastic process
is always $\Tplus \times \Omega$, this will not be a source of confusion.}:
\begin{equation}\label{eq:shiftrevisited}
(\rho _s(u))_t(\theta_s\omega) \equiv
u_{t+s}(\omega)\,.
\end{equation}
This follows from (I4) with $v = \rho_s(u)$, which then
yields $u\lozenge_sv = u$.
\mybox
\end{remark}

The shift operator $\rho_s$ has a physical interpretation.
The right-hand side is the input as interpreted by an observer of the RDSI
$\varphi$ who started at time $t_1 = 0$, while the left-hand side is how someone
who started observing the system at time $t_2 = s$ would describe it at time $t$
($+\ t_2$). Following this interpretation, a $\theta$-stationary input would
then
be an input which is observed to be just the same, regardless of when one
started observing it.

We also introduce a notion of outputs.
\bd{def:RDSIO}\label{def:rdsio}
A \emph{random dynamical system with inputs and outputs {\em (}RDSIO\,{\em )}}
is a quadruple
$(\theta, \varphi, \calU, h)$ such that $(\theta ,\varphi, \calU)$ is an RDSI,
and
\[
h: \Omega \times X \rightarrow Y
\]
is an ($\calF\otimes\borel$)-measurable map into a metric space $Y$ such that
$h(\omega,\cdot)$ is continuous for each $\omega \in \Omega$. In this context we
call $h$ an {\em output function} and $Y$ an {\em output space}.

It may sometimes be useful to refer to a {\em random dynamical system with
outputs {\em (}RDSO\,{\em )}} only, by which we mean a triple
$(\theta,\varphi,h)$ where
$(\theta,\varphi)$ is an RDS and $h$ is an output function.
\ed
The $\Omega$-component in the domain of output functions is important. It allows
for the concept to model uncertainties in the readout as well. We will return to
systems with outputs further down, in the context of RDSIO's which can be
realized as cascades of RDSO's and RDSIO's.

\begin{example}(Linear Example)\label{ex:linear}
Suppose that $\Time = \r$, and also that $X = U = \r$. Let $\calU := \soU
\subseteq
\sthetau$ be the set of $\theta$-inputs consisting of all $\theta$-stochastic
processes $u \in \sthetau$ such that
\[
t \longmapsto u_t(\omega)\,,\quad t \geqslant 0\,,
\]
is locally essentially bounded for each $\omega \in \Omega$. For each $u
\in
\calU = \soU$, consider the linear Random
Differential Equation with Inputs (RDEI)
\begin{equation}\label{eq:linear}
 \dot \xi = a(\theta_t\omega)\xi + b(\theta_t\omega)u_t(\omega)\,,\quad t
\geqslant 0\,,
\end{equation}
where $a,b \in X^\Omega_\borel$ are such that
\[
 t \longmapsto a(\theta_t\omega)\quad \text{and}\quad t \longmapsto
b(\theta_t\omega)\,,\quad t \geqslant 0\,,
\]
are locally essentially bounded for all $\omega \in \Omega$. Then for each
$\omega \in \Omega$ and any initial state $x \in X$, the ordinary differential
equation (\ref{eq:linear}) has a unique solution $\varphi(\cdot\,;\omega,x,u)$
defined by
\ifonecolumn{%
\begin{equation}\label{eq:linearformula}
\varphi(t;\omega,x,u) := x\e^{\int_0^t
a(\theta_\tau\omega)\,d\tau} + \int_0^t
b(\theta_\sigma\omega)u_\sigma(\omega)\e^{\int_\sigma^t
a(\theta_\tau\omega)\,d\tau}\,d\sigma\,,
\end{equation}%
for each $t \geqslant 0$.
}%
\iftwocolumn{%
$t\varphi(t;\omega,x,u)=$
\begin{equation}\label{eq:linearformula}
x\e^{\int_0^t
a(\theta_\tau\omega)\,d\tau} + \int_0^t
b(\theta_\sigma\omega)u_\sigma(\omega)\e^{\int_\sigma^t
a(\theta_\tau\omega)\,d\tau}\,d\sigma\,,\quad t \geqslant 0\,.
\end{equation}%
}%
This defines an RDSI $\varphi \colon \Tplus \times \Omega \times X \times
\calU \rightarrow X$ over $\theta$. Indeed, (I1) and (I1$^\prime$)
follow
from the fact that the
limit of a sequence of measurable functions is measurable. properties (I2) and
(I3)
follow directly from (\ref{eq:linearformula}). And (I4) and (I5) follow from
uniqueness
of solutions applied for each fixed $\omega \in \Omega$---one basically verifies
that both sides of each equation we want to prove to be true, when looked
at as functions of $t$, define solutions of the same differential equation with
the same initial condition.
\mydiamond
\end{example}

We use $\linftyofX$ instead of the more traditional 
$L^\infty(\Omega;X)$ to denote the space of essentially bounded, measurable maps
$\Omega \rightarrow X$. This way we can be more consistent
with the
notations $\omegaxb$ for random variables and $\omegaxtheta$ for tempered
random variables $\Omega \rightarrow X$ (see definition in Section
\ref{sec:characteristic}). Similarly, we denote by $\sinftyU$ the
subset of
$\theta$-inputs consisting of the globally essentially bounded
$\theta$-stochastic processes in $\sthetau$.

\begin{definition}(Bounded RDSI)\label{def:boundedcocycle}
  An RDSI $(\theta,\varphi,\calU)$ is said to be {\em bounded} if the random
state
\[
  \omega \longmapsto \varphi(t,\omega,x(\omega),u)\,,\quad \omega \in \Omega\,,
\]
is essentially bounded for every fixed $(t, x, u) \in \Tplus \times \linftyofX
\times (\calU \cap \sinftyU)$.
\mytriangle
\end{definition}
We emphasize the fact that we only test the condition in the definition above
for bounded random initial states $x \in \linftyofX$.
\subsection{Pullback trajectories}\label{sec:pullbacktrajectories}
Let $(\theta,\varphi,\calU,h)$ be an RDSIO with output space $Y$. Given $x \in
\omegaqb$ and $u \in
\calU$, we define the ({\em forward}\,) {\em trajectory starting at $x$ and
subject to $u$} to be the $\theta$-stochastic process $\xi^{x,u} \in \stochQ$
defined by
\[
 \xi_t^{x,u}(\omega) := \varphi(t,\omega,x(\omega),u)\,, \quad (t,\omega) \in
\Tplus \times \Omega\,.
\]
We then define the {\em pullback trajectory starting at $x$ and subject to $u$}
to be the $\theta$-stochastic process $\check\xi^{x,u} \in \stochQ$ defined by
\[
 \check\xi_t^{x,u}(\omega) := \xi_t^{x,u}(\thetaminus{t}\omega) =
\varphi(t,\thetaminus{t}\omega,x(\thetaminus{t}\omega),u)\,, \quad (t,\omega)
\in \Tplus \times \Omega\,.
\]
The ({\em forward}\,) {\em output trajectory corresponding to initial state $x$
and input $u$} is defined to be the $\theta$-stochastic process
$\eta^{x,u} \in \sthetay$, where
\[
 \eta_t^{x,u}(\omega) := h(\theta_t\omega, \varphi(t,\omega,x(\omega),u)) =
h(\theta_t\omega,\xi_t^{x,u}(\omega))\,, \quad (t,\omega) \in \Tplus \times
\Omega\,,
\]
while the {\em pullback output trajectory corresponding to initial state $x$ and
input $u$} is analogously defined to be the $\theta$-stochastic process
$\check\eta^{x,u} \in \sthetay$, where
\[
\begin{array}{rcll}
 \check\eta_t^{x,u}(\omega) 
&:= &\eta_t^{x,u}(\thetaminus{t}\omega) &\\
&= &h(\omega,
\varphi(t,\thetaminus{t}\omega,x(\thetaminus{t}\omega),u)) &\\
&= &h(\omega, \check\xi_t^{x,u}(\omega))\,, &\quad (t,\omega) \in
\Tplus \times \Omega\,.
\end{array}
\]
For RDSI's the definitions of forward and pullback trajectories are the same and
we also use the notations $\xi^{x,u}$ and $\check\xi^{x,u}$. For RDSO's the
definitions are analogous, except that they of course do not depend on any
inputs. So forward and pullback trajectories are defined as for RDS's and we
also use the notations $\xi^x$ and $\check\xi^x$, respectively. We denote the
forward and pullback output trajectories corresponding to
initial state $x$ by $\eta^{x}$ and $\check\eta^{x}$, respectively:
\[
 \eta_t^{x}(\omega) := h(\theta_t\omega, \varphi(t,\omega,x(\omega))) =
h(\theta_t\omega, \xi^x_t(\omega))
\]
and
\[
 \check\eta_t^{x}(\omega) := h(\omega,
\varphi(t,\thetaminus{t}\omega,x(\thetaminus{t}\omega))) =
h(\omega, \check\xi^x_t(\omega))
\]
for every $(t,\omega) \in \Tplus \times \Omega$.

Note that the input $u$ is not shifted in the argument of $\varphi$ in the
pullback, while at first one might intuitively think it should have been. There
are several reasons why this is so. First notice that
\[
 \check\xi_t^{x,u}(\omega) = \xi_t^{x,u}(\thetaminus{t}\omega)\,, \quad \forall
(t,\omega) \in \Tplus \times \Omega\,.
\]
So $\check\xi^{x,u}$ is just the pullback of the $\theta$-stochastic process
$\xi^{x,u}$, as it should be the case. However we are more concerned with what
happens in the context of cascades and feedback interconnections of RDSIO's. But
before we get to that we first discuss discrete RDSIO's. This will further
motivate axioms (I1)--(I5) in the definition of an RDSI, provide ---and
completely characterize--- a whole class of examples, and provide the framework
for said discussion of pullback trajectories and cascades.

We say that an RDSI (or RDSIO) is {\em discrete} when $\Time = \z$. We first
note
that, just like RDS's \cite[Section 2.1]{arnold-2010}, RDSI's also have their
flows completely determined by
their state at time $t = 1$. 

\begin{theorem}\label{thm:discretegenerators}
 For every discrete RDSI $(\theta,\varphi,\calU)$, there exists a unique map
$f\colon \Omega \times X \times U \rightarrow X$ such that
\begin{itemize}
 \item[{\em (G1)}] $f\colon \Omega \times X \times U \rightarrow
X$
is $(\calF\otimes\borel\otimes\borel(U))$-measurable,

\item[{\em (G2)}] $f(\omega,\cdot,\tilde u)\colon X \rightarrow
X$ is continuous
for each $(\omega,\tilde u) \in \Omega \times U$,
\end{itemize}
and
\begin{equation}\label{eq:discrete}
 \varphi(n+1, \omega, x, u) = f(\theta_n\omega, \varphi(n,\omega,x,u),
u_n(\omega))\,,
\end{equation}
for every $(n,\omega,x,u) \in \Tplus \times \Omega \times X \times \calU$.

Conversely, given an MPDS $\theta$, a set of $\theta$-inputs $\calU$ and a map
$$f\colon \Omega \times X \times U \rightarrow X$$ satisfying {\em (G1)} and
{\em
(G2)}, define $\varphi\colon \Tplus \times \Omega \times X \times \calU
\rightarrow X$ recursively by
\begin{equation}\label{eq:discretedef1}
 \varphi(0,\omega,x,u) := x\,,\quad (\omega,x,u) \in \Omega \times X \times
\calU\,,
\end{equation}
and {\em Equation (\ref{eq:discrete})}. Then $(\theta,\varphi,\calU)$ is an
RDSI.

We refer to the map $f$ as the {\em generator} of the RDSI
$(\theta,\varphi,\calU)$.
\end{theorem}

\begin{proof}
Define $f$ by setting
\[
 f(\omega,x,\tilde u) := \varphi(1,\omega,x,c(\tilde u))\,,\quad
(\omega,x,\tilde u) \in \Omega \times X \times U\,.
\]
Then (G1) and (G2) follow directly from (I1$^\prime$) and (I2), respectively.
Equation (\ref{eq:discrete}) follows from (I4) (see Remark \ref{rem:rhoshift})
and (I5):
\[
\begin{array}{rcl}
\varphi(n+1,\omega,x,u) 
&= &\varphi(1,\theta_n\omega,\varphi(n,\omega,x,u),\rho_n(u)) \\[1ex]
&=
&\varphi(1,\theta_n\omega,\varphi(n,\omega,x,u),
c((\rho_n(u))_0(\theta_n\omega))) \\[1ex]
&= &f(\theta_n\omega,\varphi(n,\omega,x,u),(\rho_n(u))_0(\theta_n\omega))
\\[1ex]
&= &f(\theta_n\omega, \varphi(n,\omega,x,u),u_n(\omega))
\end{array}
\]
for any $(n,\omega,x,u) \in \Tplus \times \Omega \times X \times \calU$.
Uniqueness
follows from (I3) and (I5), together with the computations above performed
backwards
for $t = 0$.

Now suppose $f$ satisfies (G1) and (G2), and that $\varphi$ is defined
recursively by (\ref{eq:discretedef1}) and
(\ref{eq:discrete}). For (I1), pick any $u \in \calU$.
One first shows using induction on $n$ that
\begin{equation}\label{eq:inductivestep}
\varphi(n,\cdot,\cdot,u) =
f(\theta_{n-1}\cdot,\varphi(n-1,\cdot,\cdot,u),u_{ n-1 } (\cdot))
\end{equation}
is (${\cal F} \otimes \borel$)-measurable for each $n \in \z_{>0}$. Indeed,
at $n = 1$ we have 
\[
\varphi(1,\cdot,\cdot,u) =
f(\theta_{1-1}\cdot,\varphi(1-1,\cdot,\cdot,u),u_{1-1} (\cdot)) = f(\cdot,
\cdot, u_0(\cdot))\,,
\]
which is (${\cal F} \otimes \borel$)-measurable, since $f$ satisfies (G1) and
$u_0$ is $\calF$-measurable. Now Equation (\ref{eq:inductivestep}) gives us the
inductive step, since the righthand side is a composition of measurable
functions and, hence, itself measurable. Now pick any $A \in \borel$. We then
have
\[
\varphi(\cdot,\cdot,\cdot,u)^{-1}(A) 
=
\displaystyle\bigcup_{n=0}^\infty \{n\} \times \varphi(n,\cdot,\cdot,u)^{-1}(A)
\in
2^{\z_{\geqslant 0}} \otimes {\cal F} \otimes \borel\,,
\]
since it is a countable union of ($2^{\z_{\geqslant 0}} \otimes {\cal F} \otimes
\borel$)-measurable sets. Thus (I1) holds. One can prove (I1$^\prime$) in the
same way by noting that
\[
\widetilde\varphi^{-1}(A) =
\displaystyle\bigcup_{n=0}^\infty \{n\} \times
\widetilde\varphi(n,\cdot,\cdot,\cdot)^{-1}(A)
\]
for each $A \in \borel$, and that
\[
\widetilde\varphi(n,\cdot,\cdot,\cdot) =
f(\theta_{n-1}\cdot,\widetilde\varphi(n-1,\cdot,\cdot,\cdot),\cdot)
\]
is $(\calF\otimes\borel\otimes\borel(U))$-measurable for
each $n \in \z_{> 0}$.

Property (I2) follows from (G2), (\ref{eq:discretedef1}) and
(\ref{eq:discrete}), again by induction on $n \in \z_{\geqslant 0}$. Indeed, at
$n=0$, $\varphi(0,\omega,\cdot,u)$ is continuous for every $\omega \in \Omega$
and every $u \in \calU$. So once (I2) has been proved for a certain value of $n
\in \z_{\geqslant 0}$, we conclude that
\[
\varphi(n+1,\omega,\cdot,u) = f(\theta_n\omega, \varphi(n,\omega,\cdot,u),
u_n(\omega))
\]
is continuous for any $\omega \in \Omega$ and any $u \in \calU$ as well.

Property (I3) follows from Equation (\ref{eq:discretedef1}).

Before proving (I4) we first prove (I5) by
induction on $n \in \Z_{\geqslant 0}$. Fix $\omega \in \Omega$, $x \in X$.
Equation (\ref{eq:discretedef1}) gives us the base of
the induction. Now assume (I5) holds for a certain value of $n \in
\z_{\geqslant 0}$. If $u,v \in \calU$ are such that $u_j(\omega) =
v_j(\omega)$ for $j =
0,1,\ldots,n$, then $\varphi(n,\omega,x,u) =
\varphi(n,\omega,x,v)$ by the induction hypothesis. So it follows from
(\ref{eq:discrete}) that
\[
\begin{array}{rcl}
\varphi(n+1,\omega,x,u) &= &f(\theta_n\omega,\varphi(n,\omega,x,u),u_n(\omega))
\\
&= &f(\theta_n\omega,\varphi(n,\omega,x,v),v_n(\omega)) \\
&= &\varphi(n+1,\omega,x,v)\,.
\end{array}
\]
This proves (I5).

It remains to prove (I4). For each arbitrarily fixed $p \in \Z_{\geqslant 0}$,
we use induction on $n \in \Z_{\geqslant 0}$. For $n = 0$, (I4) holds in virtue
of (I3) and (I5). For any $\omega \in \Omega$, we have $u_j(\omega) =
(u\lozenge_pv)_j(\omega)$ for $j = 0,\ldots,p-1$. Therefore
\[
  \varphi(0,\theta_p\omega,\varphi(p,\omega,x,u),v) = \varphi(p,\omega,x,u) =
\varphi(0+p,\omega,x,u\lozenge_pv)\,,
\]
for any $x \in X$. Now suppose (I4) holds for some $n \in \z_{\geqslant 0}$.
Given $\omega \in \Omega$ and $x \in X$, set $y :=
\varphi(n,\theta_p\omega,x,u)$. Then
\[
\begin{array}{rcl}
\varphi(n + 1,\theta_p \omega, y, v) &= & f(\theta_{n}\theta_p\omega,
\varphi(n,\theta_p\omega,y,v),v_{n}(\theta_p\omega))\\[1ex]
&=
&f(\theta_{n+p}\omega,\varphi(n+p,\omega,x,u\lozenge_pv),
(u\lozenge_pv)_{n+p}(\omega))\\[1ex]
&= &\varphi(n+p+1,\omega,x,u\lozenge_pv)\,.
\end{array}
\] 
This completes the proof that $(\theta,\varphi,\calU)$ is an RDSI.
\end{proof}

Observe that we did not need (I1) in order to prove the first half of the
theorem. So we could have in principle dropped this axiom from the definition
of an RDSI and an analogous result would still hold. We remind the reader that
(I1) was nevertheless used in showing that RDSI's restricted to
$\theta$-stationary inputs are RDS's (see Lemma \ref{lemm:rdsi-constant}
below).

From the construction of the generator $f$ of an RDSI $(\theta,\varphi,\calU)$,
it is clear how the dependence of the flow $\varphi$ at time $n \in
\z_{\geqslant 0}$ and subject to $\omega \in \Omega$ on the input $u$ is really
through the value $u_n(\omega)$ of the input $u$. So when one shifts the
$\Omega$-argument $\omega$ of $\varphi$ in the pullback trajectory to
$\thetaminus{n}\omega$, there is no need to change the input, since
$\varphi(n,\thetaminus{n}\omega,x(\thetaminus{n}\omega),u)$ depends on
$u_n(\thetaminus{n}\omega)$ already. This is our second reason for defining the
pullback trajectories of systems with inputs like so.

We now discuss the third and most important reason this is the
mathematically sensible way of defining pullback trajectories for RDSI's. Let
$(\theta,\psi)$ be a discrete RDS evolving on the state space $Z =
X_1 \times X_2$:
\[
  \psi\colon \z_{\geqslant 0} \times \Omega \times (X_1 \times X_2)
\longrightarrow
(X_1 \times X_2)\,.
\]
Let $g\colon \Omega \times Z \rightarrow Z$ be the generator of
$(\theta,\psi)$. Suppose $g$ can be written as
\begin{equation}\label{eq:discretegeneratorcascade}
  g(\omega,(x_1,x_2)) \equiv
  \begin{pmatrix}
f_1(\omega,x_1)
\\ f_2(\omega,x_2, h_1(\omega,x_1))
  \end{pmatrix}\,,
\end{equation}
where $f_1\colon \Omega \times X_1 \rightarrow X_1$ is the generator of some
RDSO $(\theta,\varphi_1,h_1)$ with output space $Y_1$, and $f_2\colon \Omega
\times X_2 \times U_2 \rightarrow X_2$ is the generator of some RDSI
$(\theta,\varphi_2,\calU_2)$ with input space $U_2 = Y_1$. Let $\pi_2\colon X_1
\times X_2 \rightarrow X_2$ be the projection onto the second coordinate. We
use $\eta_1$ to denote the output trajectories of $(\theta,\varphi_1,h_1)$,
$\xi$
for the state trajectories of $\psi$, and $\xi_2$ for the state trajectories
of $(\theta,\varphi_2,\calU_2)$.

\begin{theorem}\label{prop:discretecascades}
For any random initial state $$z = (x_1,x_2) \in Z_{\borel(Z)}^\Omega =
(X_1)_{\borel(X_1)}^\Omega \times (X_2)_{\borel(X_2)}^\Omega\,,$$
the following two identities hold:
\begin{itemize}
  \item[{\em(1)}] $\psi(n,\omega,z(\omega))
\equiv \begin{pmatrix}                                      
  \varphi_1(n,\omega,x_1(\omega)) \\
  \varphi_2(n,\omega,x_2(\omega),(\eta_1)^{x_1})
\end{pmatrix}$, and

  \item[{\em(2)}] $\pi_2(\check\xi_n^{z}(\omega)) \equiv
(\check\xi_2)_n^{x_2,(\eta_1)^{x_1}}(\omega)$.
\end{itemize}
\end{theorem}

\begin{proof}
(1) For each fixed $\omega \in \Omega$ and $z \in Z_{\borel(Z)}^\Omega$, we use
induction on $n \in \z_{\geqslant 0}$. At $n = 0$ we have
\[
  \psi(0, \omega, z(\omega)) = z(\omega) =
\begin{pmatrix}
x_1(\omega) \\ x_2(\omega)  
\end{pmatrix} =
\begin{pmatrix}
\varphi_1(0, \omega, x_1(\omega)) \\ \varphi_2(0, \omega, x_2(\omega),
(\eta_1)^{x_1})
\end{pmatrix}\,.
\]
Now suppose that (1) holds for some $n \in \z_{\geqslant 0}$. Since
$$h_1(\theta_{n}\omega, \varphi_1(n, \omega,
x_1(\omega))) = (\eta_1)_{n}^{x_1}(\omega)$$ by definition, it follows that
\[
  \begin{array}{rcl}
    \psi(n+1, \omega, z(\omega)) 
&= & g(\theta_{n}\omega, \psi(n, \omega, z(\omega)))
\\[1ex]
&= & \begin{pmatrix}
       f_1(\theta_{n}\omega, \varphi_1(n, \omega,
x_1(\omega))) \\ 
       f_2(\theta_{n}\omega, \varphi_2(n, \omega, x_2(\omega),
(\eta_1)^{x_1}), (\eta_1)_{n}^{x_1}(\omega) )
     \end{pmatrix} \\[2.5ex]
&= & \begin{pmatrix}
       \varphi_1(n + 1, \omega, x_1(\omega)) \\ 
       \varphi_2(n + 1, \omega, x_2(\omega),
(\eta_1)^{x_1})
     \end{pmatrix}\,.
  \end{array}
\]
This completes the induction.

(2) We prove by induction that (2) holds, for each $n \in \z_{\geqslant 0}$, for
all random initial states $z = (x_1,x_2) \in Z_{\borel(Z)}^\Omega$, and all
$\omega \in \Omega$. At $n = 0$ we have
\[
\begin{array}{rcl}
  \pi_2(\check\xi_0^{z}(\omega))
&= &\pi_2(\psi(0,\omega,(x_1(\omega),x_2(\omega)))) \\[1ex]
&= &x_2(\omega) \\[1ex]
&= &\varphi_2(0,\omega,x_2(\omega),(\eta_1)^{x_1}) \\[1ex]
&= &(\check\xi_2)_0^{x_2,(\eta_1)^{x_1}}\,. 
\end{array}
\]
Now assume (2) has been proved to hold for all integer values of $n$ up to some
$n_0 \geqslant 0$, for all random initial states $z = (x_1,x_2) \in
Z_{\borel(Z)}^\Omega$ and all $\omega \in \Omega$. Given $z = (x_1,x_2) \in
Z_{\borel(Z)}^\Omega$, define $\hat z = (\hat x_1, \hat x_2) \in
Z_{\borel(Z)}^\Omega$ by
\begin{equation}\label{eq:zhat}
\begin{array}{rcl}
  \hat z(\omega) 
&= &
  g(\thetaminus{1}\omega, z(\thetaminus{1}\omega)) \\[1ex]
&:= &
\begin{pmatrix}
  f_1(\thetaminus{1}\omega, x_1(\thetaminus{1}\omega)) \\
  f_2(\thetaminus{1}\omega, x_2(\thetaminus{1}\omega), h_1(\thetaminus{1}\omega,
x_1(\thetaminus{1}\omega)))
\end{pmatrix}\,, \quad \omega \in \Omega\,.
\end{array}
\end{equation}
We have $(\eta_1)^{\hat x_1} = \rho_1((\eta_1)^{x_1})$ by Lemma
\ref{lemma:zhat}
below, and also
\[
h_1(\thetaminus{(n_0+1)}\omega, x_1(\thetaminus{(n_0+1)}\omega)) = 
(\eta_1)^{x_1}_0(\thetaminus{(n_0+1)}\omega)\,, \quad \omega \in \Omega\,.
\]
Fix $\omega \in \Omega$ arbitrarily and denote $\hat\omega :=
\thetaminus{(n_0+1)}\omega$. Then
\[
\begin{array}{rcl}
  \pi_2(\check\xi_{n_0+1}^{z}(\omega))
&= &\pi_2(\psi(n_0+1, \hat\omega,
z(\hat\omega))) \\[1ex]
&= &\pi_2(\psi(n_0, \thetaminus{n_0}\omega,
\psi(1, \hat\omega, z(\hat\omega)))) \\[1ex]
&= &\pi_2(\psi(n_0, \thetaminus{n_0}\omega, g(\hat\omega,
z(\hat\omega)))) \\[1ex]
&= &\pi_2(\psi(n_0, \thetaminus{n_0}\omega, \hat z(\thetaminus{n_0}\omega)))
\\[1ex]
&= &\pi_2(\check\xi_{n_0}^{\hat z}(\omega)) \\[1ex]
&= &(\check\xi_2)_{n_0}^{\hat x_2, (\eta_1)^{\hat x_1}}(\omega)
\end{array}
\]
by the induction hypothesis. Now
\[
\begin{array}{rcl}
(\check\xi_2)_{n_0}^{\hat x_2, (\eta_1)^{\hat x_1}}(\omega)
&= &\varphi_2(n_0, \thetaminus{n_0}\omega, \hat x_2(\thetaminus{n_0}\omega),
(\eta_1)^{\hat x_1}) \\[1ex]
&= &\varphi_2(n_0, \thetaminus{n_0}\omega, f_2(\hat\omega,
x_2(\hat\omega),(\eta_1)_0^{x_1}(\hat\omega)),
(\eta_1)^{\hat x_1}) \\[1ex]
&= &\varphi_2(n_0, \thetaminus{n_0}\omega,
\varphi_2(1, \hat\omega,
x_2(\hat\omega),(\eta_1)^{x_1}),
\rho_1((\eta_1)^{\hat x_1})) \\[1ex]
&= &\varphi_2(n_0+1, \thetaminus{(n_0+1)}\omega,
x_2(\thetaminus{(n_0+1)}\omega),(\eta_1)^{x_1}) \\[1ex]
&= &(\check\xi_2)_{n_0+1}^{x_2, (\eta_1)^{x_1}}(\omega)\,.
\end{array}
\]
So
\[
\pi_2(\check\xi_{n_0+1}^{z}(\omega)) =  (\check\xi_2)_{n_0+1}^{x_2,
(\eta_1)^{x_1}}(\omega)\,.
\]
Since $z = (x_1,x_2) \in Z_{\borel(Z)}^\Omega$ and $\omega \in \Omega$ were
arbitrary, this completes the inductive step.
\end{proof}
The lefthand side of (2) in the proposition above is the projection over the
second coordinate of the pullback trajectory starting at $z = (x_1,x_2)$ of the
RDS $(\theta,\psi)$. The righthand side is the pullback trajectory of
the RDSI $(\theta,\varphi_2,\calU_2)$ starting at $x_2$ and subject to the
input $(\eta_1)^{x_1}$, the output trajectory of $(\theta,\varphi_1,h_1)$
starting at $x_1$. Theorem \ref{prop:discretecascades} then says that they
coincide. An analogous result holds in continuous time for systems generated by
random differential equations. These provide the motivation for the definition
of cascades of systems with inputs and outpus, an introductory discussion of
which is carried out in Subsection \ref{subsec:cascades}.

We now state and prove the technical lemma referred to in the proof of item (2)
in Theorem \ref{prop:discretecascades}:

\begin{lemma}\label{lemma:zhat}
Let $f\colon \Omega \times X \rightarrow X$ be the generator of a discrete RDSO
$(\theta,\varphi,h)$. Given $x \in \omegaqb$, let $\hat x
\in \omegaqb$ be defined by
\[
  \hat x(\omega) := f(\thetaminus{1}\omega,x(\thetaminus{1}\omega))\,, \quad
\omega \in \Omega\,.
\]
Then $\eta^{\hat x} = \rho_1(\eta^x)$.
\end{lemma}

\begin{proof}
Indeed, we have
\[
\begin{array}{rcl}
  \eta^{\hat x}_n(\omega)
  &= & h(\theta_n\omega, \varphi(n,\omega,\hat x(\omega))) \\[1ex]
  &= & h(\theta_n\omega, \varphi(n,\omega,f(\theta_{-1}\omega,
x(\theta_{-1}\omega)))) \\[1ex]
  &= & h(\theta_n\omega,
\varphi(n,\omega,\varphi(1,\theta_{-1}\omega,x(\theta_{-1}\omega)))) \\[1ex]
  &= & h(\theta_{n+1}\theta_{-1}\omega, \varphi(n+1, \theta_{-1}\omega,
x(\theta_{-1}\omega))) \\[1ex]
  &= & \eta_{n+1}^x(\theta_{-1}\omega) \\[1ex]
  &= & (\rho_1(\eta^x))_{n}(\omega)\,,
\end{array}
\]
for every $n \in \z_{\geqslant 0}$ and every $\omega \in \Omega$.
\end{proof}
\subsection{$\theta$-Stationary Inputs}
The concept of RDSI subsumes that of an RDS, as we shall see
below. Denote the subset of $\sthetau$ consisting of $\theta$-stationary inputs
by
$\sthetauconst$. We identify $\sthetauconst$ and $\omegaub$ via Lemma
\ref{lemma:theta-constant}. 

Let $(\theta ,\varphi, \calU)$ be a RDSI, and
suppose that $\bar u
\in
\calU \cap \sthetauconst$ is some
$\theta $-stationary input. Consistent with the convention that an overbar is
used to indicate the $\theta$-stationary process associated with a given random
variable, we remove the bar to denote the random variable associated with a
given $\theta$-stationary process. So we denote by $u$ the random variable in
$\omegaub$ associated via Lemma \ref{lemma:theta-constant} with $\bar u$.
We then define
\[
\varphi_u := \varphi(\cdot,\cdot,\cdot,\bar u) \colon \Tplus \times \Omega
\times
X \longrightarrow X\,.
\]

\bl{lemm:rdsi-constant}
$\varphi_u$ is a crude cocycle.
\els

\bpr
It follows from condition (I1) and \cite[Proposition 2.34,
page 65]{folland-1999} that $\varphi_u$ is measurable. From (I2),
$\varphi_u(t,\omega,\cdot)$ is continuous for each $(t,\omega) \in \Tplus
\times \Omega$, yielding (S1). From (I3), we know that $\varphi_u(0, \omega,
\cdot) = id_X$ for
every $\omega \in \Omega$. So to verify (S2$^\prime$) it remains to prove that
$\varphi_u$ satisfies the
``crude cocycle property''. Let $\widetilde{\Omega} \subseteq \Omega$ be a
$\theta$-invariant subset of full measure such that
\begin{equation}\label{eq:rdsi-constant}
 (\rho_s(\bar u))_t(\omega) = \bar u_t(\omega)\,,\quad \forall s,t \in
\Tplus\,,\ \forall
\omega \in \widetilde{\Omega}\,.
\end{equation}
Fix arbitrarily $\omega \in \widetilde{\Omega}$. For any $s,t \in \Tplus$, we
have $\theta_s\omega \in \widetilde{\Omega}$ by $\theta$-invariance, and so it
follows from (\ref{eq:rdsi-constant}) and (I5) that
\[
 \varphi(t,\theta_s\omega,\varphi_u(s,\omega,x),\rho_s(\bar u)) =
\varphi(t,\theta_s\omega,\varphi_u(s,\omega,x),\bar u)\,.
\]
It then follows from (I4)---see Remark \ref{rem:rhoshift}---that
\beqn
\varphi_u(t+s,\omega ,x) &=& \varphi(t+s,\omega ,x,\bar u)\\
               &=& \varphi(t,\theta _s\omega ,\varphi(s,\omega ,x,\bar u),\rho
_s(\bar u))\\
               &=& \varphi(t,\theta _s\omega ,\varphi_u(s,\omega ,x),\bar u)\\
               &=& \varphi_u(t,\theta _s\omega ,\varphi_u(s,\omega ,x))\,.
\eeqn
So (S2$^\prime$) is satisfied with $\Omega_s := \widetilde\Omega$ for every $s
\in \Tplus$.
\epr
\ifarxiv{
\begin{proposition}\label{prop:perfection}
 $\varphi_u$ can be perfected.
\end{proposition}

\begin{proof}
 Let $\widetilde{\Omega}$ be the $\theta$-invariant subset of full measure of
$\Omega$ from the proof of Lemma \ref{lemm:rdsi-constant}, and denote $N :=
\Omega\backslash\widetilde{\Omega}$. Note that $N$ is $\theta$-invariant and
$\realp(N) = 0$. Define $\psi_u\colon \Tplus \times \Omega \times X \rightarrow
X$ by
\begin{equation}\label{eq:psiu}
 \psi_u(t,\omega,x) := \left\{
\begin{array}{rl}
 \varphi_u(t,\omega,x)\,,\ &\text{if}\ \omega \in \widetilde{\Omega}\,, \\
 x\,,\ &\text{if}\ \omega \in N\,.
\end{array}
\right.
\end{equation}
Pick any open subset $A \subseteq X$. Then
\[
 \psi_u^{-1}(A) = \big(\varphi_u^{-1}(A) \backslash \Tplus \times N \times
X\big) \cup \Tplus \times N \times A \,.
\]
Thus $\psi_u^{-1}(A) \in \borel(\Tplus) \otimes \calF \otimes \borel$, proving
that $\psi_u$ is ($\borel(\Tplus) \otimes \calF \otimes \borel$)-measurable.

If $\omega \in \widetilde{\Omega}$, then $\psi_u(t,\omega,\cdot) =
\varphi_u(t,\omega,\cdot) = \varphi(t,\omega,\cdot,u)$, which is continuous for
any $t \geqslant 0$ by (I2). And if $\omega \in N$, then $\psi_u(t,\omega,\cdot)
= id_X$, and thus also continuous for any $t \geqslant 0$. This shows $\psi_u$
satisfies (S1).

It is clear from (\ref{eq:psiu}) that $\psi_u(0,\omega,\cdot) = id_X$ for any
$\omega \in \Omega$. We already know from the proof of Lemma
\ref{lemm:rdsi-constant} that $\psi_u$ satisfies the cocycle property for every
$\omega \in \widetilde\Omega$. Similar computations using (\ref{eq:psiu})
together with the fact that $N$ is also $\theta$-invariant show that $\psi_u$
also satisfies the cocycle property for values of $\omega$ in $N$. This proves
(S2), completing the proof that $(\theta,\psi_u)$ is an RDS.

Since $\realp(N) = 0$, we conclude that $\varphi_u$ and $\psi_u$ are
indistinguishable. Hence $\psi_u$ is a perfection of $\varphi_u$.
\end{proof}
}

Whenever the state space $X$ is such that $\varphi_u$ can be
perfected, we shall assume that $\varphi_u$ has already been replaced by an
indistinguishable perfection and then refer to the resulting RDS
$(\theta,\varphi_u)$.
\subsection{Input to State Characteristics}\label{sec:characteristic}
\bd{def:characteristic}
Let $(\theta ,\varphi, \calU)$ be an RDSI, and suppose that $\bar u \in
\calU \cap \sthetauconst$, with corresponding random variable $u$ (see Lemma
\ref{lemma:theta-constant}).
An \emph{equilibrium associated to $\bar{u }$} (or to $u$) is any
equilibrium 
${\xi }$ of the RDS $(\theta ,\varphi_{u})$.
The set of all equilibria associated to $\bar{u }$ is denoted as
$\Eq(\bar{u })$ (or $\Eq(u)$). 
\ed
In other words, an element ${x }\in \Eq(\bar{u })$ is a random variable
$\Omega \rightarrow X$
such that 
\be{eq:equil_equiv_char}
\varphi_u(t,\thetaminus{t}\omega
,x(\thetaminus{t}\omega )) =
x(\omega )\,,
\quad
\forall t\in \Tplus, \,
\widetilde{\forall} \omega \in \Omega \,.
\ee
When we have a ``proper'' RDS (no inputs), we write simply $\Eq$ for the set of
equilibria. 

Though not really used in this work, we take advantage of the concepts and
notation being introduced to present the following quick observation:

\begin{proposition}\label{prop:output-equil}
Let $(\theta,\varphi,h)$ be an RDSIO and suppose that ${\xi }\in \Eq({\mu })$.
Then the output trajectory $\eta^{\xi} \colon \Tplus \times \Omega \rightarrow
Y$ starting at $\xi$ is $\theta $-stationary.
\end{proposition}
\bpr
Indeed, it follows from (\ref{eq:equil_equiv_char}) that, for any $t,s \geqslant
0$ and any $\omega \in \Omega$,
\beqn
(\rho _s(\eta^\xi ))_t(\omega ) &=& \eta^\xi_{t+s}(\theta _{-s}\omega )\\ 
\ifonecolumn{&=&}\iftwocolumn{&&\adjusttoleftmore =\;\;}
 h(\theta _{t+s}\theta _{-s}\omega,{\varphi(t+s,\theta _{-s}\omega ,{\xi
}(\theta_{-s}\omega),\bar{\mu})})\\
\ifonecolumn{&=&}\iftwocolumn{&&\adjusttoleftmore =\;\;}
h(\theta _{t}\omega,{\varphi(t+s,\theta _{-(s+t)}\theta_t\omega
,{\xi }(\theta _{-(s+t)}\theta_t\omega),\bar{\mu})})\\
\ifonecolumn{&=&}\iftwocolumn{&&\adjusttoleftmore =\;\;}
h(\theta _t\omega,{{\xi }(\theta _{t}\omega )})\\
\ifonecolumn{&=&}\iftwocolumn{&&\adjusttoleftmore =\;\;}
h(\theta_t\omega,{\varphi(t, \omega, {\xi}(\omega), \bar\mu)})\\
\ifonecolumn{&=&}\iftwocolumn{&&\adjusttoleftmore =\;\;}
\eta^\xi_t(\omega)\,,
\eeqn
proving that $\rho_s(\eta^\xi) = \eta^\xi$ for all $s \in \Tplus$, that is,
$\eta^\xi $ is $\theta $-stationary.
\epr

\bd{def:iocharacteristic}
Let $(\theta ,\varphi,h)$ be an RDSIO, and suppose that ${\xi }\in
\Eq({\mu })$.
Then $\eta_0^\xi$ is an \emph{output equilibrium associated to $\mu$}
(or $\bar\mu $). The set of all output equilibria associated to ${\mu }$ is
denoted as $\Oeq({\mu })$. 
\ed

For deterministic systems ($\Omega $ is a singleton), when the set $\Eq(\bar{u
})$ consists of a single globally attracting equilibrium, the mapping
${u }\mapsto \Eq(\bar{u })$  is the object called the
``input to state characteristic'' in the literature on monotone i/o systems.
For systems with outputs, composition with the output map $h$ provides the
``input to output'' characteristic \cite{monotoneTAC}. We extend this notion to
RDSI's. For
reasons which will be illustrated in Example \ref{ex:lineari2schar} and become
clearer in the proofs of Theorems
\ref{thm:randomcics} and \ref{thm:boundedcics}
(converging input to converging state), further conditions on the
convergence of the states are needed. 

In what follows, given an MPDS $\theta$ and a normed space $(X,\|\cdot\|)$, we
denote by $\omegaxtheta$ \label{page:temperedrv} the space of tempered random
variables $\Omega
\rightarrow X$; that is, the space of $\calF$-measurable maps $r\colon \Omega
\rightarrow X$ such that
\[
  \sup_{s \in \Time} \|r(\theta_s\omega)\|e^{-\gamma|s|} < \infty\,,\quad
\forall\gamma > 0\,,\ \widetilde{\forall}\omega \in \Omega\,.
\]
For ease of reference we note a few easy properties of tempered random
variables:
\begin{lemma}\label{lemma:temperedproduct}
  Suppose $\theta$ is an MPDS, $(X,\|\cdot\|)$ is a normed space over $\r$, and
let $R_1,R_2 \in X^\Omega_\theta$, $r \in \r_\theta^\Omega$, and $c
\in \r$. Then
\begin{itemize}
\item[{\em (1)}] $R_1 + R_2$ is tempered.

\item[{\em (2)}] $cR_1$ is tempered.

\item[{\em (3)}] $rR_1$ is tempered; in particular, the product of two
real-valued tempered random variables is tempered.
\end{itemize}
\end{lemma}

\begin{proof}
(1) Indeed, for any $\gamma > 0$ and any $\omega \in
\widetilde{\Omega}$, we have
\[
\begin{array}{rcl}
\displaystyle\sup_{s \in \Time} \|(R_1 + R_2)(\theta_s\omega)\|e^{-\gamma|s|}
&\leqslant &\displaystyle\sup_{s \in \Time}
\|R_1(\theta_s\omega)\|e^{-\gamma|s|} + \sup_{s \in \Time}
\|R_2(\theta_s\omega)\|e^{-\gamma|s|} \\[2ex]
&< &\infty\,,
\end{array}
\]
where we write $(R_1+R_2)(\theta_s\omega)$ for $R_1(\theta_s\omega) +
R_2(\theta_s\omega)$. 
So both $R_1 + R_2$ is tempered.

(2) follows from (3), which we now prove. Given $\gamma > 0$ and $\omega \in
\widetilde{\Omega}$, apply the
definition of tempered random variable for $\gamma/2$:
\[
\begin{array}{rcl}
\displaystyle\sup_{s \in \Time}
\|r(\theta_s\omega)R_1(\theta_s\omega)\|e^{-\gamma|s|}
&= &\displaystyle\sup_{s \in \Time} |r(\theta_s\omega)|e^{-\frac{\gamma}{2}|s|}
\|R_1(\theta_s\omega)\|e^{-\frac{\gamma}{2}|s|} \\[2ex]
&\leqslant &\displaystyle\left(\sup_{s \in \Time}
|r(\theta_s\omega)|e^{-\frac{\gamma}{2}|s|}\right)
 \left(\sup_{s \in \Time} \|R_1(\theta_s\omega)\|e^{-\frac{\gamma}{2}|s|}\right)
\\[2ex]
&< &\infty\,.
\end{array}
\]
Thus $rR_1$ is tempered.
\end{proof}
In other words, $X^\Omega_\theta$ is a real vector space, and a module over the
ring of real-valued tempered random variables.

\begin{definition}(Tempered Convergence)\label{def:temperedconvergence}
 Let $\theta$ be an MPDS, $X$ be a normed space, $(\xi_\alpha)_{\alpha \in
A}$ be a
 net in $\omegaxb$ and ${\xi}_\infty$ any random variable in $\omegaxb$. We say
that $(\xi_\alpha)_{\alpha \in A}$ converges to
${\xi}_\infty$ in the {\em tempered sense} if there exists a nonnegative,
tempered
random variable $r\colon \Omega \rightarrow \r_{\geqslant 0}$ and an $\alpha_0
\in A$ such that
\begin{itemize}
 \item[(1)] $\xi_\alpha(\omega) \rightarrow {\xi}_\infty(\omega)$ as $\alpha
\rightarrow
\infty$ for $\theta$-almost all $\omega \in \Omega$, and

 \item[(2)] $\|\xi_\alpha(\omega) - {\xi}_\infty(\omega)\| \leqslant
r(\omega)$
for all
$\alpha \geqslant \alpha_0$, for $\theta$-almost all $\omega \in \Omega$.
\end{itemize}
In this case we denote $\xi_\alpha \temparrow \xi_\infty$ (as $\alpha
\rightarrow \infty$).
\mytriangle
\end{definition}

\begin{definition}(Tempered Continuity)\label{def:i2scharcontinuity}
 Let $\theta$ be an MPDS and $X,U$ normed spaces. A map ${\cal K}\colon
\calU \subseteq \omegaub
\rightarrow \omegaxb$ is said do be {\em tempered continuous} if, whenever
$(u_\alpha)_{\alpha \in A}$ is a net in $\calU$
convergent to $u_\infty \in \calU$ in the tempered sense, then
$\calK(u_\alpha) \temparrow \calK(u_\infty)$ as $\alpha \rightarrow \infty$ as
well.
\mytriangle
\end{definition}

In what follows, when we speak of a given RDS $(\theta,\varphi)$, or of an
RDSI $(\theta,\varphi,\calU)$, etc, the underlying state space $X$, input
space $U$ and output space $Y$ are all assumed to be normed spaces, unless
otherwise specified.

\begin{definition}(I/S Characteristic)\label{def:i2scharacteristicproper}
  An RDSI $(\theta,\varphi,\calU)$ is said to have an {\em input to state
{\em (}i/s\,{\em)}
characteristic} $\calK\colon \omegautheta \rightarrow \omegaxtheta$ if
$$\omegautheta \subseteq \calU$$ and
$$\varphi_{u}(t,\theta_{-t}\omega,x(\thetaminus{t}\omega))
\longrightarrow ({\calK}(u))(\omega)\,, \quad \text{as} \ t \rightarrow
\infty\,,\
\widetilde{\forall} \omega \in \Omega\,,$$ for each $u \in \omegautheta$,
and each $x \in \omegaxtheta$.
\mytriangle
\end{definition}

\begin{definition}\label{def:i2scharacteristic}
Let $(\theta, \varphi, \calU)$ be an RDSI with an i/s characteristic
$\calK\colon \omegautheta \rightarrow \omegaxtheta$. We say that $\calK$ is a
{\em tempered i/s characteristic} if
\[
  \check\xi^{x,u}_t \longtemparrow \calK(u)\,,\quad \forall x \in
\omegaxtheta\,,\ \forall u \in \omegautheta\,.
\]
If 
\[
  \check\xi^{x,u}_t \longtemparrow \calK(u)\,,\quad \forall x \in \linftyofX\,,\
\forall u \in \linftyofU\,,
\]
then $\calK$ is said to be a {\em bounded i/s characteristic}.
\mytriangle
\end{definition}
Notice that $\linftyofU \subseteq \omegautheta$,
so that the definition of bounded i/s characteristic is well-posed.
\begin{example}(Linear Example)\label{ex:lineari2schar}
 Consider the RDSI $(\theta,\varphi,\calU)$ from Example \ref{ex:linear},
generated by the random
differential equation with inputs
\begin{equation*}
 \dot \xi = a(\theta_t\omega)\xi + b(\theta_t\omega)u_t(\omega)\,,\quad t
\geqslant 0\,,\quad u \in \calU := \soU\,.
\end{equation*}
Suppose that, in addition to the
hypotheses in Example \ref{ex:linear}, $a,b$ also satisfy
\begin{itemize}
 \item[(L1)] $b$ is tempered; and

 \item[(L2)] there exist a $\lambda > 0$ and nonnegative,
tempered random
   variables $\gamma,\widetilde\gamma \in (\r_{\geqslant})_\theta^\Omega$ such
that
\[
  \e^{\int_s^{r+s} a(\theta_\tau\omega)\,d\tau} \leqslant
\gamma(\theta_s\omega)\e^{-\lambda r}\,,\quad 
\ifonecolumn{\forall s \in
\r\,,\ \forall r \geqslant 0\,,\ \widetilde{\forall} \omega \in
\Omega\,,
\]}\iftwocolumn{\]for all $s \in \r$, $\forall r\geqslant0$, for
$\theta$-almost every $\omega \in \Omega$,}
and
\[
  \e^{\int_s^{r+s} a(\theta_{-\tau}\omega)\,d\tau} \leqslant
\widetilde\gamma(\theta_s\omega)\e^{-\lambda r}\,,\quad 
\ifonecolumn{\forall s \in
\r\,,\ \forall r \geqslant 0\,,\ \widetilde{\forall} \omega \in
\Omega\,.
\]}\iftwocolumn{\]for all $s \in \r$, $\forall r\geqslant0$, for
$\theta$-almost every $\omega \in \Omega$.}
\end{itemize}
\begin{remark}
 If we assume $\theta$ to be ergodic and
\[
 {\mathbb E}a := \int_\Omega a(\omega)\,d\realp(\omega) < 0\,,
\]
then one can show (see \cite{chueshov-2002}) that (L2) holds with $\lambda :=
-{\mathbb E}a$.
\mybox
\end{remark}

We summarize the point of this example in the lemma and proposition below.

\begin{lemma}\label{lemma:continuousi2s}
 Suppose that, in addition to the hypotheses in {\em Example \ref{ex:linear}},
the
coefficients $a$ and $b$ also satisfy {\em(L1)} and {\em(L2)}.
Then the RDSI
$(\theta,\varphi,\calU)$ has a tempered continuous input to state
characteristic.
\end{lemma}

\begin{proposition}\label{prop:boundedi2s}
 Suppose that, in addition to satisfying the hypotheses in {\em Example
\ref{ex:linear}}, plus {\em (L1)} and {\em (L2)}, $b$ is also essentially
bounded, and $a$ is essentially bounded from above {\em(}as functions of $\omega
\in
\Omega$\,{\em)}. Then the RDSI $(\theta,\varphi,\calU)$ is bounded {\em(}in the
sense of {\em Definition
\ref{def:boundedcocycle}}\,{\em)}, and has a
continuous
and bounded input to state characteristic {\em(}in the sense of {\em Definitions
\ref{def:i2scharcontinuity}{\em,} \ref{def:i2scharacteristicproper} {\em and}
\ref{def:i2scharacteristic}}\,{\em)}.
\end{proposition}

\noindent We prove both these results together.

Pick any $u \in \omegautheta$. Then (L1) and (L2) imply
(\cite{chueshov-2002}) that the limit
\begin{equation}\label{eq:lineark}
\iftwocolumn{\adjusttoleft}
  \lim_{t\to\infty}
\varphi(t,\thetaminus{t}\omega,x(\thetaminus{t}\omega),\bar{u}) = 
\ifonecolumn{%
\int_{-\infty}^0 b(\theta_\sigma\omega)u(\theta_\sigma\omega)\e^{\int_\sigma^0
a(\theta_\tau\omega)\,d\tau}\,d\sigma
\end{equation}%
}%
\iftwocolumn{%
\]
\[
\int_{-\infty}^0 b(\theta_\sigma\omega)u(\theta_\sigma\omega)\e^{\int_\sigma^0
a(\theta_\tau\omega)\,d\tau}\,d\sigma
\]%
}%
exists for $\theta$-almost every $\omega \in \Omega$, for any tempered initial
state $x \in \omegaxtheta$. Let
$\calK(u)\colon \Omega \rightarrow X$
be the map defined $\theta$-almost everywhere in $\Omega$ by
\begin{equation}\label{eq:lineari2schar}
  (\calK(u))(\omega) := \int_{-\infty}^0
b(\theta_\sigma\omega)u(\theta_\sigma\omega)\e^{\int_\sigma^0
a(\theta_\tau\omega)\,d\tau}\,d\sigma\,.
\end{equation}
By Equation (\ref{eq:lineark}) and Proposition \ref{prop:convergencetoeq},
$\calK(u)$ is an equilibrium of $(\theta,\varphi_u)$.

{\bf $\calK(u)$ is tempered.} Fix $\delta > 0$ arbitrarily. By (L2), for each $s
\in \r$,
we have
\[
  \begin{array}{rcl}
    |(\calK(u))(\theta_s\omega)|\e^{-\delta|s|} 
\ifonecolumn{&\leqslant&}\iftwocolumn{\\[2ex]&&\adjusttoleftmore
\leqslant\;\;}\displaystyle
\int_{-\infty}^0|b(\theta_{\sigma+s}\omega)u(\theta_{\sigma+s}
\omega)\gamma(\theta_{\sigma+s}\omega)|\e^{-(\lambda|\sigma|+\delta|s|)}\,
d\sigma \\[2ex]
\ifonecolumn{&\leqslant&}\iftwocolumn{&&\adjusttoleftmore
\leqslant\;\;}\displaystyle
\int_{-\infty}^0|b(\theta_{\sigma+s}\omega)u(\theta_{\sigma+s}
\omega)\gamma(\theta_{\sigma+s}\omega)|\e^{-m|\sigma + s|}\,
d\sigma \\[2ex]
\ifonecolumn{&\leqslant&}\iftwocolumn{&&\adjusttoleftmore
\leqslant\;\;}\displaystyle
K_{\omega,\delta}\int_{-\infty}^0 \e^{-m|\sigma + s|/2}\,d\sigma \\[2ex]
\ifonecolumn{&\leqslant&}\iftwocolumn{&&\adjusttoleftmore
\leqslant\;\;}\displaystyle
K_{\omega,\delta}\int_{-\infty}^\infty \e^{-m|\sigma|/2}\,d\sigma\,,
  \end{array}
\]
which is finite and depends only on $\omega$ and $\delta$---in the
computations above we used $m := \min\{\lambda,\delta\} > 0$ and
\[
  K_{\omega,\delta} := \sup_{s \in \r}
|b(\theta_{s}\omega)u(\theta_{s}
\omega)\gamma(\theta_{s}\omega)|\e^{-m|s|/2} < \infty\,.
\]
$K_{\omega,\delta}$ being finite follows from the hypotheses that $b,u,\gamma$
are tempered and the fact that the product of tempered random variables is also
tempered (Lemma \ref{lemma:temperedproduct} (3)).

Observe that $u \in U_\theta^\Omega$ was chosen arbitrarily. Therefore this
shows that (\ref{eq:lineari2schar}) defines an input to state
characteristic 
$\calK\colon \omegautheta \rightarrow \omegaxtheta$ for
$(\theta,\varphi,\calU)$. We now show that $\calK$ is tempered continuous in the
sense of Definition
\ref{def:i2scharcontinuity}, and bounded in the sense of Definition
\ref{def:i2scharacteristic}.

{\bf Continuity.} Suppose that $u \in \calU$ converges to $u_\infty \in
\omegautheta$ in the tempered sense, with corresponding $r \in
(\r_{\geqslant 0})_\theta^\Omega$, and $\alpha = t_0 \geqslant 0$ (see
Definition
\ref{def:temperedconvergence}). Then $\calK({\check u_t}) \temparrow
\calK({u_\infty})$ as well. Indeed, for $t \geqslant
t_0$, we have
\[
  \begin{array}{rcl}
\iftwocolumn{\adjusttoleftmore&&}
|\calK({\check u_t})(\omega) - \calK({u_\infty})(\omega)| 
\ifonecolumn{&\leqslant &\displaystyle}
\iftwocolumn{\\[2ex]&\leqslant &\displaystyle}
\int_{-\infty}^0|b(\theta_{\sigma}\omega)|
|u_t(\theta_{\sigma}\omega) - u_\infty(\theta_\sigma\omega)|\e^{\int_\sigma^0
a(\theta_\tau\omega)\,d\tau}\,d\sigma \\[2ex]
&\leqslant &\displaystyle\int_{-\infty}^0|b(\theta_{\sigma}\omega)|
r(\theta_{\sigma}\omega)\gamma(\theta_{\sigma}\omega)\e^{-\lambda|\sigma|}\,
d\sigma \\[2ex]
&=: &R(\omega)\,,\quad \widetilde{\forall} \omega \in \Omega\,.
  \end{array}
\]
Now computations along the lines of the ones above showing that $\calK(u)$ is
tempered will show that the random variable $R$ so defined is also
tempered.

This completes the proof of Lemma \ref{lemma:continuousi2s}. We now discuss
Proposition \ref{prop:boundedi2s}.

{\bf $\calK$ is bounded.} Fix $u \in \linftyofU$ arbitrarily. If $x \in
\linftyofX$, then
\[
\begin{array}{rcl}
 |x(\thetaminus{t}\omega)\e^{\int_0^t a(\theta_{\tau-t}\omega)\,d\tau} - 0|
&= &|x(\thetaminus{t}\omega)\e^{\int_{0}^t a(\theta_{-\tau}\omega)\,d\tau}|
\\[2ex]
&\leqslant &\|x\|_\infty\widetilde{\gamma}(\omega)e^{-\lambda t}\\[2ex]
&\leqslant &\|x\|_\infty\widetilde{\gamma}(\omega)\,,
\ifonecolumn{\quad \forall t \geqslant 0\,,\ \widetilde{\forall} \omega \in
  \Omega\,.\end{array}\]}\iftwocolumn{\end{array}\]$\forall t \geqslant 0\,,\
\widetilde{\forall} \omega \in  \Omega$.}
Furthermore,
\[
 \begin{array}{rl}
   &\hspace{-4em}
\left| \displaystyle\int_0^t
b(\theta_{\sigma-t}\omega)u(\theta_{\sigma-t}\omega)\e^{\int_\sigma^t
a(\theta_{\tau-t}\omega)\,d\tau}\,d\sigma \right. \\
&\hspace{-1em}
\ifonecolumn{-}\iftwocolumn{\right.\\[2ex]&-\left.}
\left.\displaystyle\int_{-\infty}^0
b(\theta_{\sigma}\omega)u(\theta_{\sigma}\omega)\e^{\int_\sigma^0
a(\theta_{\tau}\omega)\,d\tau}\,d\sigma \right| \\[2ex]
&\hspace{4em} = \left|\displaystyle\int_{-\infty}^{-t}
b(\theta_{\sigma}\omega)u(\theta_{\sigma}\omega)\e^{\int_\sigma^0
a(\theta_{\tau}\omega)\,d\tau}\,d\sigma\right|\\[2ex]
&\hspace{4em} \leqslant \displaystyle\int_{-\infty}^{0}
|b(\theta_{\sigma}\omega)u(\theta_{\sigma}\omega)|\gamma(\theta_\sigma\omega)\e^
{-\lambda|\sigma|}\,d\sigma\\[2ex]
&\hspace{4em} =: R_1(\omega)\,, \quad\widetilde{\forall}\omega \in \Omega\,,
 \end{array}
\]
the random variable $R_1$ so defined being tempered as per computations such as
above. This shows that $\calK(\check\xi_t^{x,u}) \temparrow \calK(u)$ as $t
\rightarrow \infty$, for every $x \in \linftyofX$, and for every $u \in
\linftyofU$. Therefore $\calK$ is
bounded.

Finally, we show that

{\bf $\varphi$ is bounded.} Now suppose that $b$ is essentially bounded and, in
addition to satisfying (L2), $a$ is bounded from above almost everywhere by a
constant $M$.
Then for any $x \in \linftyofX$,
\[
\begin{array}{rcl}
  |x(\thetaminus{t}\omega)\e^{\int_0^t a(\theta_{\tau-t}\omega)\,d\tau}|
&= &|x(\thetaminus{t}\omega)\e^{\int_{-t}^0 a(\theta_{\tau}\omega)\,d\tau}|
\\[2ex]
\ifonecolumn{&\leqslant&}
\iftwocolumn{&&\adjusttoleftmore\leqslant\;\;}
\|x\|_\infty\e^{Mt}\,,\quad \forall t \geqslant 0\,,\
\widetilde{\forall} \omega \in \Omega\,.
\end{array}
\]
And if $u \in \calU \cap \sinftyU$, then
\[
\left| \displaystyle \int_0^t
b(\theta_{\sigma-t}\omega)u_{\sigma-t}(\omega)\e^{\int_\sigma^t
a(\theta_{\tau-t}\omega)\,d\tau}\,d\sigma \right| 
\ifonecolumn{\leqslant}\iftwocolumn{\] \[ \leqslant}
\|b\|_\infty\|u\|_\infty \displaystyle \int_0^t \e^{M(t-\sigma)}
\,d\sigma < \infty\,,
\]
for all $t \geqslant 0$ and $\theta$-almost all $\omega \in \Omega$.
This shows that $\varphi$ is bounded in the sense of Definition
\ref{def:boundedcocycle}. This completes the proof of Proposition
\ref{prop:boundedi2s}.
\mydiamond
\end{example}
\section{Monotone RDSI's}
\newcommand{\temperednn}{(\r_{\geqslant 0})^\Omega_\borel}

If $(X,\leqslant )$ is a partially ordered space and $p,q\in \stochQ$, we write
$p\leqslant q$
to mean that $p(t,\omega )\leqslant q(t,\omega )$ for all $t\in \Tplus$ and all
$\omega \in \Omega $.

\bd{def:monotoneRDIOS}
An RDSI $(\theta ,\varphi,\calU)$ is said to be \emph{monotone}
if the underlying state and input spaces are partially ordered $(X,
\leqslant_X)$, $(U, \leqslant_U)$, and $$\varphi(t,\omega,x,u) \leqslant_X 
\varphi(t,\omega,z,v)$$ whenever $x,z \in X$ and $u,v \in \stochU$ are such
that $x\leqslant_X z$ and $u\leqslant_U v$
\ed
Most often the underlying partial order will be clear from the context and we
shall use simply ``$\leqslant$'' to denote either of ``$\leqslant_X$'',
``$\leqslant_U$'' or ``$\leqslant_Y$''.

\subsection{Converging Input to Converging State}

Recall the convention of using a check mark $\check{\phantom{u}}$ above the
symbol for a given $\theta$-stochastic process to denote its corresponding
pullback flow.

\begin{theorem}[Tempered CICS]\label{thm:randomcics}
 Let $(\theta,\varphi,\calU)$ be a monotone RDSI with state space $X = \r^n$ and
input space $U = \r^k$, both equipped with the usual positive orthant-induced
partial order. Suppose that $\varphi$ has a continuous and tempered i/s
characteristic ${\cal K}\colon \temperedU \rightarrow \temperedX$. If $u \in
\calU$ and $u_\infty \in \temperedU$ are such that
\begin{equation}\label{eq:inputconvergence}
 \check{u}_t \longrightarrow_\theta u_\infty\,,
\end{equation}
then
\[
 \check\xi^{x,u}_t \longrightarrow_\theta \calK(u_\infty)\,,\quad \forall x \in
\temperedX\,.
\]
In other words, if the pullback trajectory of $u$ converges to $u_\infty$ in
the tempered sense, then the pullback trajectories of $\varphi$ subject to the
given input $u$ and starting at any tempered random state $x$ will converge to
$\calK(u_\infty)$ as well.
\end{theorem}

\begin{proof}
For each $\tau \geqslant 0$, let $a_\tau,b_\tau \in \omegaub$ be
defined by
\begin{equation}\label{eq:cicsproofinf}
 a_\tau(\omega) := \inf_{t\geqslant\tau}u_t(\thetaminus{t}\omega)
\end{equation}
and
\begin{equation}\label{eq:cicsproofsup}
 b_\tau(\omega) := \sup_{t\geqslant\tau}u_t(\thetaminus{t}\omega)
\end{equation}
for each $\omega\in\Omega$, where the $\inf$ and $\sup$ are taken
coordinatewise. It follows from (\ref{eq:inputconvergence}) that
\[
 a_\tau(\omega),b_\tau(\omega) \longrightarrow u_\infty(\omega)\,,\as \tau
\rightarrow \infty\,,\quad \widetilde{\forall}\omega \in \Omega\,.
\]
 Let $t_0 \geqslant 0$ and $r \in (\r_{\geqslant 0})_\theta^\Omega$ be such that
\[
 |u_t(\thetaminus{t}\omega) - u_\infty(\omega)| \leqslant r(\omega)\,,\quad
\forall \tau
\geqslant t_0\,,\ \widetilde{\forall}\omega \in \Omega\,.
\]
So it follows from the continuity of the Euclidean norm in $\r^n$ that
\[
 |a_\tau(\omega) - u_\infty(\omega)|,|b_\tau(\omega) - u_\infty(\omega)|
\leqslant r(\omega)\,,\quad \forall \tau
\geqslant t_0\,,\ \widetilde{\forall}\omega \in \Omega\,.
\]
also. In other words, $a_\tau,b_\tau \longtemparrow u_\infty$ as $\tau
\rightarrow
\infty$. Moreover,
\[
 a_{\tau}(\theta_t\omega) \leqslant u_t(\omega) \leqslant
b_{\tau}(\theta_t\omega)\,, \quad t \geqslant \tau \geqslant 0\,,\quad \omega
\in \Omega.
\]

For each $\tau \geqslant 0$, let $\bar a_\tau, \bar b_\tau$ be the
$\theta$-stationary stochastic processes generated by $a_\tau,b_\tau$,
respectively. Then
\[
 (\bar a_\tau)_s(\omega) = a_\tau(\theta_s\omega) =
 a_\tau(\theta_{s+\tau}\thetaminus{\tau}\omega)
\iftwocolumn{\]\[}
 \leqslant u_{\tau + s}( \thetaminus{\tau}\omega) = (\rho_\tau(u))_s(\omega)\,,
\]
and, similarly,
\[
 (\rho_\tau(u))_s(\omega) \leqslant (\bar b_\tau)_s(\omega)\,, \quad s,\tau
\geqslant 0\,,\quad \omega \in \Omega\,.
\]
Thus
\begin{equation}\label{eq:ataurhoubtau}
 \bar a_\tau \leqslant \rho_\tau(u) \leqslant \bar b_\tau\,,\quad \tau \geqslant
0\,.
\end{equation}

Fix arbitrarily $x_0 \in \omegaxb$. For any $\omega \in \Omega$ and any $t
\geqslant \tau \geqslant 0$, we have
\[
 \begin{array}{rcl}
  |\check\xi^{x_0,u}_t(\omega) - (\calK(u_\infty))(\omega)|
&\leqslant &|\check\xi^{x_0,u}_t(\omega) - \check\xi^{x_0,\bar
a_\tau}_t(\omega)| \\[1ex]
&&\iftwocolumn{\adjusttoleft}+ |\check\xi^{x_0,\bar a_\tau}_t(\omega) -
(\calK(a_\tau))(\omega)| \\[1ex]
&&\iftwocolumn{\adjusttoleft}+
|(\calK(a_\tau))(\omega)-(\calK(u_\infty))(\omega)|\,.
 \end{array}
\]
We first show that
\begin{equation}\label{eq:star}
 |\check\xi^{x_0,u}_t(\omega) - (\calK(u_\infty))(\omega)| \longrightarrow
0\,,\as t \rightarrow \infty\,,\quad \widetilde{\forall} \omega \in \Omega\,.
\end{equation}

Given $\epsilon > 0$, it follows from the continuity of $\calK$ that there
exists $\tau_0 \geqslant 0$ such that
\[
 |(\calK(a_\tau))(\omega)-(\calK(u_\infty))(\omega)|
\iftwocolumn{,\] \[}
\ifonecolumn{,}
 |(\calK(b_\tau))(\omega)-(\calK(u_\infty))(\omega)| < \epsilon/6\,,\quad \tau
\geqslant \tau_0\,.
\]
Now we can use the definition of input to state characteristic to choose $t_0
\geqslant \tau_0$ such that
\[
 |\check\xi^{x_0,\bar a_{\tau_0}}_t(\omega) - (\calK(a_{\tau_0}))(\omega)| <
\epsilon/6\,,\quad t \geqslant t_0\,.
\]
Using the cocycle property we may write
\[
\begin{array}{rcl}
\iftwocolumn{&&}
 \check\xi^{x_0,u}_t(\omega) 
\iftwocolumn{\\[1ex]}
&= &
\varphi(t-\tau_0,\thetaminus{(t-\tau_0)}\omega,\varphi(\tau_0,\thetaminus{t}
\omega,{x}_0(\thetaminus{t}\omega),u),\rho_{\tau_0}(u)) \\[1ex]
 &=&
\varphi(t-\tau_0,\thetaminus{(t-\tau_0)}\omega,x_1(\thetaminus{(t-\tau_0)}
\omega),\rho_{\tau_0}(u)) \\[1ex]
 &=& \varphi(s,\thetaminus{s}\omega,x_1(\thetaminus{s}\omega),\rho_{\tau_0}(u))
\\[1ex]
 &=& \check\xi^{x_1,\rho_{\tau_0}(u)}_s(\omega)\,,
\end{array}
\]
where $x_1 \in \omegaxb$ is defined by $x_1(\omega) :=
\varphi(\tau_0,\theta_{\tau_0}\omega,x_0(\theta_{\tau_0}\omega),u)$, $\omega \in
\Omega$, and we write $s := t - \tau_0$. Now by monotonicity
\[
 \check\xi^{x_1,\bar a_{\tau_0}}_s(\omega) 
 \leqslant \check\xi^{x_1,\rho_{\tau_0}(u)}_s(\omega)
 \leqslant \check\xi^{x_1,\bar b_{\tau_0}}_s(\omega)\,,
\]
hence, using that $x \leqslant y \leqslant z$ in the positive orthant order
implies $|y-x| \leqslant |z-x|$, we have
\[
\begin{array}{rcl}
\iftwocolumn{&&}
|\check\xi^{x_1,\rho_{\tau_0}(u)}_s(\omega)
 - \check\xi^{x_1,\bar a_{\tau_0}}_s(\omega)|
\iftwocolumn{\\[1ex]}
&\leqslant& |\check\xi^{x_1,\bar b_{\tau_0}}_s(\omega)
 - \check\xi^{x_1,\bar a_{\tau_0}}_s(\omega)|\\[1ex]
&\leqslant& |\check\xi^{x_1,\bar b_{\tau_0}}_s(\omega)
 - (\calK(b_{\tau_0}))(\omega)| \\[1ex]
 &&+|(\calK(b_{\tau_0}))(\omega) - (\calK(u_0))(\omega)| \\[1ex]
 &&+|(\calK(u_0))(\omega) - (\calK(a_{\tau_0}))(\omega)| \\[1ex]
 &&+|(\calK(a_{\tau_0}))(\omega) - \check\xi^{x_1,\bar a_{\tau_0}}_s(\omega)|
\\[1ex]
 &\leqslant& |\check\xi^{x_1,\bar b_{\tau_0}}_s(\omega) -
(\calK(b_{\tau_0}))(\omega)| \\[1ex]
&& +|(\calK(a_{\tau_0}))(\omega) - \check\xi^{x_1,\bar a_{\tau_0}}_s(\omega)|
 +\epsilon/3\,,
\end{array}
\]
for every $s \geqslant 0$.
Again from the definition of input to state characteristic, one can choose $s_0
\geqslant 0$ large enough so that
\[
 |\check\xi^{x_1,\bar b_{\tau_0}}_s(\omega) - (\calK(b_{\tau_0}))(\omega)|, 
 |(\calK(a_{\tau_0}))(\omega) - \check\xi^{x_1,\bar a_{\tau_0}}_s(\omega)|
 < \epsilon/6\,,\ifonecolumn{\quad s \geqslant s_0\,.\]}\iftwocolumn{\]$s
\geqslant s_0$.}
It then follows that
\[
 |\check\xi^{x_0,u}_t(\omega) - (\calK(u_\infty))(\omega)| < \epsilon/2 +
\epsilon/3 + \epsilon/3 = \epsilon\,,\ifonecolumn{\quad t \geqslant
\max\{t_0,\tau_0+s_0\}\,.\]}\iftwocolumn{\]$t \geqslant
\max\{t_0,\tau_0+s_0\}$.}
And since $\epsilon > 0$ was arbitrary, (\ref{eq:star}) holds.

It remains to show that the convergence occurs in the tempered sense. Indeed,
pick $\tau_1,\tau_2,t_3,s_4,s_5 \geqslant 0$ and $r_1,r_2,r_3,r_4,r_5 \in
\temperednn$ such that
\[
 |(\calK(a_\tau))(\omega)-(\calK(u_\infty))(\omega)| \leqslant
r_1(\omega)\,,\quad \tau \geqslant \tau_1\,,\quad \widetilde{\forall}\omega \in
\Omega\,,
\]
\[
 |(\calK(b_\tau))(\omega)-(\calK(u_\infty))(\omega)| \leqslant
r_2(\omega)\,,\quad \tau \geqslant \tau_2\,,\quad \widetilde{\forall}\omega \in
\Omega\,,
\]
\[
|\check\xi^{x_0,a_{\widetilde{\tau}_0}}_t(\omega) -
(\calK(a_{\widetilde{\tau}_0}))(\omega)| \leqslant r_3(\omega)\,,\quad t
\geqslant t_3\,,\quad \widetilde{\forall}\omega \in \Omega\,,
\]
where $\widetilde{\tau}_0 := \max\{\tau_1,\tau_2\}$, and where we may assume
that $t_3 \geqslant \widetilde{\tau_0}$,
\[
 |\check\xi^{x_2,b_{\widetilde{\tau}_0}}_s(\omega) -
(\calK(b_{\widetilde{\tau}_0}))(\omega)| \leqslant r_4(\omega)\,,\quad s
\geqslant s_4\,,\quad \widetilde{\forall}\omega \in \Omega\,,
\]
and
\[
 |(\calK(a_{\widetilde{\tau}}))(\omega) -
\check\xi^{x_2,a_{\widetilde{\tau}_0}}_s(\omega)| \leqslant r_5(\omega)\,,\quad
s \geqslant s_5\,,\quad \widetilde{\forall}\omega \in \Omega\,,
\]
where $x_2 \in \omegaxb$ is defined by $x_2(\omega) :=
\varphi(\widetilde{\tau}_0,\theta_{\widetilde{\tau}_0}\omega,x_0(\theta_{
\widetilde{\tau}_0}\omega))$, $\omega \in \Omega$. Then similar estimates as the
ones we just carried out for the almost everywhere pointwise convergence give us
\[
 |\check\xi^{x_0,u}_s(\omega) - (\calK(u_\infty))(\omega)| 
\iftwocolumn{\]\[}
\leqslant r_1(\omega) + r_3(\omega) + (r_4(\omega) + r_2(\omega) + r_1(\omega) +
r_5(\omega))\,,
\]
for all
\[
 t \geqslant T_0 := \max\{t_3,\widetilde{\tau}_0 + s_4, \widetilde{\tau}_0 +
s_5\}\,,
\]
and for almost every $\omega \in \Omega$. Since the sum of tempered random
variables is still a tempered random variable---Lemma
\ref{lemma:temperedproduct} (1)---, the theorem follows by setting $R
:= 2r_1 + r_2 + r_3 + r_4 + r_5$.
\end{proof}

\begin{theorem}[Bounded CICS]\label{thm:boundedcics}
 Let $(\theta,\varphi,\calU)$ be a monotone RDSI with state space $X = \r^n$ and
input
space $U = \r^k$, both equipped with the usual positive orthant-induced partial
order. Suppose that $\varphi$ is bounded---see {\em Definition
\ref{def:i2scharacteristic}}---and has a continuous and bounded i/s
characteristic ${\cal K}\colon \temperedU \rightarrow \temperedX$. If $u \in
\calU \cap S_\infty^U$ and $u_\infty \in \boundedU$ are such that
\[
 \check{u}_t \longrightarrow_\theta u_\infty\,,
\]
then
\[
 \check\xi^{x,u}_t \longrightarrow_\theta \calK(u_\infty)\,,\quad \forall x \in
\boundedX\,.
\]
In other words, if the pullback trajectory of $u$ converges to $u_\infty$ in
the tempered sense, then the pullback trajectories of $\varphi$ subject to the
given input $u$ and starting at any bounded random state $x$ will converge to
$\calK(u_\infty)$.
\end{theorem}

\begin{proof}
 The proof is essentially the same as the proof of Theorem \ref{thm:randomcics}
with just a couple of extra observations. The assumption that $\varphi$ is
bounded will cause the random variable $x_2$ defined along said proof to be also
bounded. So pullback trajectories starting at $x_2$ and subject to a
$\theta$-stationary input will converge to the appropriate state characteristic.
Moreover, the hypothesis that $u \in
\calU \cap S_\infty^U$ implies that the random variables $a_\tau, b_\tau$ are
also bounded. So
\[
 \check\xi^{x,a_\tau}_t \longtemparrow \calK(a_\tau)
\]
and
\[
 \check\xi^{x,b_\tau}_t \longtemparrow \calK(b_\tau)
\]
as $t \rightarrow \infty$ for any $x \in \linftyofX$, for any $\tau \geqslant
0$. And so the estimates at the end of the proof of temperedness of the
convergence
\[
 \check\xi^{x,u}_t \longrightarrow_\theta \calK(u_\infty)
\]
in Theorem \ref{thm:randomcics} still hold. The result follows.
\end{proof}

Other than the tempered convergence and continuity notions in the hypotheses,
the key element in the proofs of Theorems \ref{thm:randomcics} and
\ref{thm:boundedcics} is the existence and measurability of the $\inf$'s and
$\sup$'s in Equations (\ref{eq:cicsproofinf}) and (\ref{eq:cicsproofsup}).
Therefore these results can be generalized to a wider class of spaces
and orders satisfying enough geometric conditions for that to happen.
\subsection{Cascades}\label{subsec:cascades}

We now discuss a few corollaries of the `converging input to converging state'
theorems just proved. Separate work in preparation deals with a small-gain
theorem for random dynamical systems.

Let $(\theta,\psi)$ be an autonomous RDS evolving on a space $Z = X_1 \times
X_2$. We say that $(\theta,\psi)$ is {\em cascaded} if the flow $\psi$ can be
decomposed as
\[
  \psi(t,\omega,(x_1(\omega),x_2(\omega))) \equiv
\begin{pmatrix}
  \varphi_1(t,\omega,x_1(\omega)) \\
  \varphi_2(t,\omega,x_2(\omega),(\eta_1)^{x_1})
\end{pmatrix}\,,
\]
for some RDSO $(\theta,\varphi_1,h_1)$ with state space $X_1$ and output space
$Y_1$, and some RDSI $(\theta,\varphi_2,\calU_2)$ with state space $X_2$,
input space $U_2 = Y_1$, and set of $\theta$-inputs $\calU_2$ containing all
(forward) output trajectories of $(\theta,\varphi_1,h_1)$. In this case we
write $\psi = \varphi_1 \otimes \varphi_2$. Recall from
item (1) in Theorem \ref{prop:discretecascades} that if the generator of a
discrete RDS can be decomposed as in Equation
(\ref{eq:discretegeneratorcascade}), then the said RDS is a cascade. A similar
decomposition can be done for systems generated by random differential equations
with inputs whose generator satisfies the natural analogues of Equation
(\ref{eq:discretegeneratorcascade}).

\begin{example}[Bounded Outputs]\label{ex:4-1}
Let $(\theta,\psi) = (\theta, \varphi_1 \otimes \varphi_2)$ be a cascaded
RDS as above. Suppose that 
$(\theta,\varphi_1,h_1)$ is an RDSO evolving on a normed space $X_1$ such that
$(\theta,\varphi_1)$ has a unique, globally attracting equilibrium
$(\xi_1)_\infty
\in \omegaqb$:
\[
  (\check\xi_1)^{x_1}_t(\omega) \longrightarrow (\xi_1)_\infty(\omega)\,, \as t
\rightarrow \infty\,, \quad \widetilde\forall \omega \in \Omega\,,\ \forall x_1
\in (X_1)_\borel^\Omega\,.
\]
Suppose further that the output function
$h_1$ is {\em bounded}; that is, there exists $M \geqslant 0$ such
that
\[
  h_1(\omega,x_1) \leqslant M\,,\quad \forall x_1 \in X_1\,, \
\widetilde\forall\omega \in \Omega\,.
\]
We have
\[
\begin{array}{rl}
  (\check\eta_1)^{x_1}_t(\omega) 
&=\hspace{2pt} h_1(\omega, (\check\xi_1)^{x_1}_t(\omega)) \\
&\longrightarrow h_1(\omega,(\xi_1)_\infty(\omega))\,, \as t \rightarrow
\infty\,,
\quad \widetilde\forall \omega \in \Omega\,,\ \forall x_1 \in
(X_1)_\borel^\Omega\,.
\end{array}
\]
And since $h_1$ is bounded, the convergence is automatically tempered. We
denote $(u_2)_\infty := h_1(\cdot, (\xi_1)_\infty(\cdot))$. The hypothesis that
$h_1$ is bounded also guarantees that $(u_2)_\infty$ is tempered.

Now assume that
that $(\theta,\varphi_2,\calU_2)$ satisfies the hypotheses of Theorem
\ref{thm:randomcics}. It then follows that $(\theta,\psi)$ has a globally
attracting equilibrium:
\[
  \check\xi^z_t(\omega) \longrightarrow 
\begin{pmatrix}
(\xi_1)_\infty(\omega) \\
(\xi_2)_\infty(\omega)
\end{pmatrix}
 \as t \rightarrow \infty\,, \quad \widetilde\forall
\omega \in \Omega\,,\ \forall z \in
Z_\theta^\Omega\,,
\]
where $(\xi_2)_\infty := \calK((u_2)_\infty)$. In particular, the convergence
in the second coordinate is tempered.

For conditions guaranteeing that an RDS $(\theta,\varphi)$ would have a unique,
globally attracting equilibrium in the sense above, see \cite[Theorem
3.2]{cao--jiang-2010}. The assumption that the output is bounded is very
reasonable in biological applications, since there is often a cut off or
saturation on the
reading of the strength of a signal. 
\mytriangle
\end{example}

Before we consider the next example, we develop a stronger notion of regularity
for
output functions. We seek a property which preserves tempered
convergence, and which we could check it holds in specific examples.

\begin{definition}[Tempered Lipschitz]\label{def:temperedlipschitz}
  An output function $h\colon \Omega \times X \rightarrow Y$ is said to be
{\em tempered Lipschitz} (with respect to a given MPDS $\theta$) if there
exists a tempered random variable $L \in (\r_{\geqslant 0})_\theta^\Omega$ such
that
\[
  \|h(\omega,x_1) - h(\omega,x_2)\| \leqslant L(\omega)\|x_1 - x_2\|\,,\quad
\forall x_1,x_2 \in X\,, \ \widetilde{\forall} \omega \in \Omega\,.
\]
We refer to $L$ as a {\em Lipschitz random variable for $h$}.
\mytriangle
\end{definition}
For example, suppose that $X \subseteq \r^n$, and that $h\colon \Omega \times X
\rightarrow \r^k$ is an output function such that $h(\omega,\cdot)$ is
differentiable for all $\omega$ in a $\theta$-invariant set of full measure
$\widetilde \Omega \subseteq \Omega$. If the norm of the Jacobian with respect
to $x$,
$$\omega \longmapsto \|D_x h (\omega, \cdot)\| := \sup_{x \in X} |D_x h
(\omega,x)|\,, \quad \omega \in \Omega\,,$$
is tempered, then $h$ is tempered Lipschitz.

\begin{lemma}
Let $h\colon \Omega \times X \rightarrow Y$ be a tempered Lipschitz continuous
output function, $p \in \calS_\theta^X$ be a $\theta$-stochastic process in
$X$, and let $p_\infty \in \omegaqb$. Let $q\colon \Tplus \times \Omega
\rightarrow Y$ be the $\theta$-stochastic process in $Y$ defined by
\[
  q_t(\omega) := h(\omega, p_t(\omega))\,,\quad (t,\omega) \in \Tplus
\times \Omega\,,
\]
and $q_\infty \in Y^\Omega_{\borel(Y)}$ be the random variable in $Y$
defined by
\[
  q_\infty(\omega) := h(\omega, p_\infty(\omega))\,, \quad \omega \in \Omega\,.
\]
If $p_t \temparrow p_\infty$, then $q_t \temparrow q_\infty$.
\end{lemma}

\begin{proof}
It follows from continuity with respect to $x \in X$ that
\[
  q_t(\omega) = h(\omega, p_t(\omega)) \rightarrow h(\omega, p_\infty(\omega)) =
p_\infty(\omega)\,, \quad \widetilde\forall \omega \in \Omega\,.
\]
Because $p_t \temparrow p_\infty$, there exist $r \in (\r_{\geqslant
0})_\theta^\Omega$ and $t_0 \geqslant 0$ such that
\[
  \|p_t(\omega) - p_\infty(\omega)\| \leqslant r(\omega)\,, \quad \forall t
\geqslant t_0\,,\ 
\widetilde\forall \omega \in \Omega\,.
\]
Let $L$ be a Lipschitz random variable for $h$. Then
\[
\begin{array}{rcll}
  \|q_t(\omega) - q_\infty(\omega)\| 
&= &\|h(\omega,p_t(\omega)) -
h(\omega,p_\infty(\omega))\| & \\[1ex]
&\leqslant &L(\omega)\|p_t(\omega) - p_\infty(\omega)\| & \\[1ex]
&\leqslant &L(\omega)r(\omega)\,, & \quad \forall t \geqslant t_0\,,\
\widetilde\forall \omega \in \Omega\,.
\end{array}
\]
By item (3) in Lemma \ref{lemma:temperedproduct}, $Lr$ is tempered, which
completes the proof.
\end{proof}

Now suppose that $(\theta,\psi,\calU)$ is an RDSI evolving on a state space $Z
= X_1 \times X_2$. In this case we say that $(\theta,\psi,\calU)$ is {\em
cascaded} if the flow $\psi$ can be decomposed as
\[
  \psi(t,\omega,(x_1(\omega),x_2(\omega)),u) \equiv
\begin{pmatrix}
  \varphi_1(t,\omega,x_1(\omega),u) \\
  \varphi_2(t,\omega,x_2(\omega),(\eta_1)^{x_1,u})
\end{pmatrix}\,,
\]
for some RDSIO $(\theta,\varphi_1,\calU_1,h_1)$ with state space $X_1$, set of
$\theta$-inputs $\calU_1 = \calU$ and output space
$Y_1$, and some RDSI $(\theta,\varphi_2,\calU_2)$ with state space $X_2$,
input space $U_2 = Y_1$, and set of $\theta$-inputs $\calU_2$ containing all
(forward) output trajectories of $(\theta,\varphi_1,\calU_1,h_1)$. In this case
we also write $\psi = \varphi_1 \otimes \varphi_2$. Item (1) in
Theorem \ref{prop:discretecascades} can be generalized to contemplate this
kind of cascades for discrete systems, as well as systems generated by
random differential equations.

\begin{example}[Tempered Lipschitz Outputs]\label{ex:4-2}
Suppose that $(\theta,\varphi_1,\calU_1)$ and $(\theta,\varphi_2,\calU_2)$ from
the decomposition above both satisfy the hypotheses of Theorem
\ref{thm:randomcics}. If the output function $h_1$ is Lipschitz continuous, then
$(\theta,\psi,\calU)$ also has the `converging input to converging state'
property;
that is, if $u \in \calU$ is such that $\check{u}_t \temparrow u_\infty$ for
some $u_\infty \in U_\theta^\Omega$, then there exists a $\xi_\infty \in
Z_\theta^\Omega$ such that 
\begin{equation}\label{eq:productcics}
\check\xi^{z,u}_t \longtemparrow \xi_\infty\,,
\quad \forall z \in Z_\theta^\Omega\,,
\end{equation}
as well.

To see this, let $\calK_1\colon (U_1)_\theta^\Omega \rightarrow
(X_1)_\theta^\Omega$ and $\calK_2\colon (U_2)_\theta^\Omega \rightarrow
(X_2)_\theta^\Omega$ be the i/s characteristics of
$(\theta,\varphi_1,\calU_1)$ and $(\theta,\varphi_2,\calU_2)$, respectively.
Fix $$z = (x_1,x_2) \in Z_\theta^\Omega = (X_1)_\theta^\Omega \times
(X_2)_\theta^\Omega$$ arbitrarily. From Theorem \ref{thm:randomcics}, we have
\[
  (\check\xi_1)_t^{x_1,u} \longtemparrow \calK_1(u_\infty)\,.
\]
Since $h_1$ is tempered Lipschitz, it follows that
\[
  (\check\eta_1)_t^{x_1,u} \longtemparrow (u_2)_\infty\,,
\]
where
\[
(u_2)_\infty := h_1(\cdot, \calK_1(u_\infty)(\cdot))\,.
\]
It follows, again from Theorem \ref{thm:randomcics}, that
\[
  (\check\xi_2)_t^{x_2, (\eta_1)^{x_1,u}} \longtemparrow
\calK_2((u_2)_\infty)\,.
\]
Hence
\[
  \check\xi_t^{z,u} 
= 
\begin{pmatrix}
  (\check\xi_1)_t^{x_1,u} \\ (\check\xi_2)_t^{x_2, (\eta_1)^{x_1,u}}
\end{pmatrix} 
\longtemparrow
\begin{pmatrix}
  \calK_1(u_\infty) \\ \calK_2((u_2)_\infty)
\end{pmatrix}\,.
\]
Since $z \in Z_\theta^\Omega$ was picked arbitrarily, this establishes
(\ref{eq:productcics}).

The procedure above can be generalized to cascades of three or more systems to
show that the `converging input to converging state' property will hold provided
that it
holds for its components---and the intermediate outputs are tempered Lipschitz.
Even if is possible to check directly that $(\theta,\varphi,\calU)$ already
satisfies the hypotheses of Theorem \ref{thm:randomcics}, it might be easier to
check them for each component.
\mytriangle  
\end{example}
\ifarxiv{%
Examples \ref{ex:4-1} and \ref{ex:4-2} illustrate how one can obtain global
convergence results for systems decomposable into cascades. In the next section
we discuss ``closed loop'' systems.
}%
\ifarxiv{
\section{Small-Gain Theorem}

\subsection{Feedback Interconnections}

We start by establishing what it means for the feedback interconnection of two
RDSIO's to be well-posed. Given two RDSIO's, we want to consider an RDS
obtained by plugging the output of one of the systems in as an input for the
other, and viceversa. Let $(\theta,\varphi_i,\calU_i,h_i)$ be a discrete RDSIO
with state space $X_i$, input space $U_i = Y_j$, and output space $Y_i = U_j$,
$i = 1,2$, $j = 2,1$. Suppose further that the set $\calU_i$ of $\theta$-inputs
of $(\theta,\varphi_i,\calU_i,h_i)$ contains all (forward) output trajectories
of $(\theta,\varphi_j,\calU_j,h_j)$, $i = 1,2$, $j = 2,1$.

\bd{def:wp}
We say that the \emph{feedback} (\emph{interconnection})
\emph{of the RDSIO's $(\theta,\varphi_1,\calU_1,h_1)$ and
$(\theta,\varphi_2,\calU_2, h_2)$ is well-posed}
provided that, for each $x_1\in X_1$ and each $x_2\in X_2$,
\begin{itemize}
\item[(1)] there exist
$\mu = \mu_{x_1,x_2} \in \calU_1$ and $\nu = \nu_{x_1,x_2} \in
\calU_2$ such that
\beqn
\nu_t(\omega) &=& h_1(\theta _t\omega ,\varphi_1(t,\omega ,x_1,\mu)),\\
\mu_t(\omega) &=& h_2(\theta _t\omega ,\varphi_2(t,\omega ,x_2,\nu)),
\eeqn
for all $t\in \Tplus$ and all $\omega \in \Omega $;

\item[(2)] if $\mu' \in \calU_1$, $\nu' \in \calU_2$ are such that
\beqn
\nu'_t(\omega) &=& h_1(\theta _t\omega ,\varphi_1(t,\omega ,x_1,\mu')),\\
\mu'_t(\omega) &=& h_2(\theta _t\omega ,\varphi_2(t,\omega ,x_2,\nu')),
\eeqn
then $\mu'_\cdot(\omega) = \mu_\cdot(\omega)$ and $\nu'_\cdot(\omega) =
\nu_\cdot(\omega)$ for $\theta$-almost all $\omega \in \Omega$.
\end{itemize}
That is, $\mu_{x_1,x_2} , \nu_{x_1,x_2}$ must exist, and be unique in a certain
sense for the feedback interconnection to be well-posed.
\ed

\noindent One can show using arguments along the lines of the constructions in
Section \ref{sec:pullbacktrajectories} that for discrete systems the feedback
interconnections are always well-posed. The same is true of RDSIO's generated
by RDSE's.

\begin{proposition}
  Suppose that the feedback interconnection of $(\theta,\varphi_1,\calU_1,h_1)$
and $(\theta,\varphi_2,\calU_2,h_2)$ is
well-posed.  Then $\psi\colon \Tplus \times \Omega \times (X_1 \times X_2)
\rightarrow X_1 \times X_2$, defined by
\[
\psi (t,\omega ,(x_1,x_2)) \,:=\;
(\varphi_1(t,\omega ,x_1,\mu_{x_1,x_2}),\varphi_2(t,\omega
,x_2,\nu_{x_1,x_2}))\,,\quad (t,\omega,(x_1,x_2)) \in \Tplus \times \Omega
\times (X_1 \times X_2)\,,
\]
satisfies the cocyle property and $(\theta ,\psi )$ is an RDS with state space
the Cartesian product $X_1 \times X_2$.
\end{proposition}

\bl{lem:equi-feedback}
Suppose that the feedback interconnection of
$(\theta,\varphi_1,\calU_1,h_1)$ and $(\theta,\varphi_2,\calU_2,h_2)$ is
well-posed.
Then, the following two properties are equivalent:
\begin{itemize}
\item[{\em(1)}]
for the feedback system, $(\bar{\xi },\bar{\zeta })\in \Eq$;
\item[{\em(2)}]
there exists a pair of $\theta $-constant inputs $\mu $ and $\nu $ such that
\[
\bar{\mu }\in \Oeq(\bar{\nu }),\;\;
\bar{\nu }\in \Oeq(\bar{\mu }),\;\;
\bar{\xi }\in \Eq(\bar{\mu })),\;\;
\bar{\zeta }\in \Eq(\bar{\nu }))\,.
\]
\end{itemize}
\els

\subsection{Small-Gain Theorem: Discussion}

A small-gain theorem for anti-monotone systems follows by the same
steps as in the deterministic case
(\cite{monotoneTAC,dcds06}).
The key idea is to assume that characteristics are continuous in our
tempered sense and that a small gain condition is satisfied.
The main theorem will be for
an anti-monotone RDSIO so that 
iterations of characteristics converge to a unique equilibrium, which
amounts (subject to mild technical conditions on cones and the space $X$)
to the requirement that
$\Oeq$ has no period two points
except for a unique equilibrium, in the sense that
there exists some $\mu _0$ so that
\[
\bar{\mu }\in \Oeq(\Oeq(\bar{\mu }))
\;\;\Rightarrow \;\;
\bar{\mu } = \bar{\mu }_0
\,.
\]
We assume, further, that the feedback connection of this system with itself is
well-posed and monotone.
Then (almost) every solution of the closed-loop system converges to
$\Eq(\bar{\mu })$.
A proof as in
\cite{monotoneTAC,dcds06}
begins by appealing to the CICS (Converging Input/Converging State) property
for monotone RDSIO's, and establishes a contraction property on the ``limsup''
and ``liminf'' (defined as in these references) of external signals.
An alternative approach is based on the idea in
\cite{enciso_smith_sontagJDE06},
in which one first proves that the
feedback of two monotone or two anti-monotone systems is monotone. 
More precisely,
if $(\theta _1,\varphi_1,\calU_1,h_1)$ and $(\theta _2,\varphi_2,\calU_2,h_2)$
are both
anti-monotone RDSIO's, we
consider the following order in $X_1\times X_2$:
\[
(x,z)\leqslant (\widetilde x,\widetilde z) \quad\mbox{iff}\quad
x\leqslant \widetilde x \;\mbox{and}\; \widetilde z\leqslant z
\]
(that is, the order on the second component is reversed).
If, instead, $(\theta _1,\varphi_1,\calU_1,h_1)$ and $(\theta _2,\varphi_2,
\calU_2,h_2)$ are
both monotone
RDSIO's, we consider the product order in $X_1\times X_2$:
\[
(x,z)\leqslant (\widetilde x,\widetilde z) \quad\mbox{iff}\quad
x\leqslant \widetilde x \;\mbox{and}\; z\leqslant \widetilde z\,.
\]
One then shows that the feedback system is monotone, and the small-gain
condition assures that there is a unique equilibrium (a.e.) for the composite
system, allowing one to appeal to the theorem in
\cite{cao--jiang-2009}.%
}%


\end{document}